\documentclass[runningheads,envcountsame]{llncs}
\newif\ifanonymous
\anonymousfalse

\usepackage[T1]{fontenc}

\usepackage{amsmath, xspace,  mathtools, multirow, calc, cite, url, lipsum, hyperref, cleveref, censor}
\renewcommand\censor[1]{#1}

\usepackage{etoolbox,hyperref,thm-restate}
\makeatletter
\pretocmd{\thmt@rst@storecounters}{\Hy@SaveLastskip}{}{}
\apptocmd{\thmt@rst@storecounters}{\Hy@RestoreLastskip}{}{}
\makeatother

\def\set#1{\ensuremath{\left\{#1\right\}}}
\def\tup#1{\ensuremath{\left\langle#1\right\rangle}}
\def\ceiling#1{\ensuremath{\left\lceil #1\right\rceil}}
\def\floor#1{\ensuremath{\left\lfloor #1\right\rfloor}}
\def\problem#1{{\rm\textsf{#1}}}

\def\dlnmarnote#1{\marginpar{\scriptsize\sloppypar\raggedright\color{red!75!black} DLN: #1}}

\def\crmarnote#1{\marginpar{\scriptsize\sloppypar\raggedright\color{green!75!black} CR: #1}}

\ifanonymous
\fi

\newcommand{\ceq}{\coloneqq}
\def\OPT{\ensuremath{\text{\rm\texttt{OPT}\xspace}}}
\def\GREEDY{{\rm\texttt{GREEDY}}\xspace}

\let\autoref\cref

\usepackage{tikz}
\usetikzlibrary{matrix,fit,positioning,calc,math}
\pgfdeclarelayer{background}
\pgfdeclarelayer{foreground}
\pgfsetlayers{background,main,foreground}
\newenvironment{flag}{\color{red!75!black}}{\color{black}}

\usepackage{color}

\begin{document}
\title{Maximizing the Margin between Desirable and Undesirable Elements in a Covering Problem}
\titlerunning{Margin between Desirable and Undesirable Elements in a Covering Problem}

\author{
\ifanonymous
\else
  Sophie Boileau\inst{1}\orcidID{0009-0000-4896-5854} \and
  Andrew Hong\inst{1,2}\orcidID{0009-0005-2772-5399} \and
  David Liben-Nowell\inst{1}\orcidID{0000-0002-9763-4303} \and
  Alistair Pattison\inst{1,3}\orcidID{0009-0008-3946-9822} \and
  Anna N.~Rafferty\inst{1}\orcidID{0000-0002-8319-5370} \and
  Charlie Roslansky\inst{1}\orcidID{0000-0002-0765-0343}
\fi
}

\authorrunning{
\ifanonymous
\else
  S.~Boileau et al.
\fi
}
\institute{
\ifanonymous
\else    
  Carleton College, Northfield, MN 55057, USA \and
  Stony Brook University, Stony Brook, NY 11794, USA \and
  Opportunity Insights, Harvard University, Cambridge, MA 02138, USA
  \email{\{boileau.sophiem,andrewhongcs,alistairpattison,roslanskyc\}@gmail.com\\
  \{dln,arafferty\}@carleton.edu}
\fi  
}

\maketitle
\begin{abstract}
  In many covering settings, it is natural to consider the presence both of elements that we seek to include and of elements that we seek to avoid.
  This paper introduces a novel combinatorial problem formalizing this tradeoff: from a collection of sets containing both ``desirable'' and ``undesirable'' items, pick the subcollection that maximizes the \emph{margin} between the number of desirable and undesirable elements covered.
  We call this the \emph{Target Approximation Problem} (TAP) and argue that many real-world scenarios are naturally modeled via this objective.
  We first show that TAP is hard, even when restricted to cases where the given sets are small or where elements appear in only a small number of sets.
  In a large swath of these cases, we show that TAP is hard even to approximate.
  We then exhibit exact polynomial-time algorithms for other restricted cases and provide an efficient 0.5-approximation for the case where elements occur at most twice, derived through a tight connection to the greedy algorithm for Unweighted Set Cover.
\end{abstract}


\section{Introduction}
\label{sec:intro}

In a well-studied class of combinatorial problems, we face a collection $\mathcal{S}$ of subsets of some groundset, and we must select some subcollection $\mathcal{S}' \subseteq \mathcal{S}$.
Typically, the measure of benefit of $\mathcal{S}'$ is the size (or weight) of the union of $\mathcal{S}'$ --- i.e., we seek to cover all or most of the groundset --- and the cost of $\mathcal{S}'$ is the number (or weight) of the chosen sets.
For example, \problem{Set Cover} requires us to cover every groundset element while minimizing the number (or weight) of chosen sets~\cite{karp1972}.
\problem{Max-$k$-Cover} requires covering as many groundset elements as possible using only $k$ sets~\cite{36582,garey-and-johnson}. There is also the complementary cost-focused formulation, \problem{Min-$k$-Union}, in which we seek to \emph{avoid} groundset elements: we must cover as \emph{few} groundset elements as possible while choosing at least $k$ sets~\cite{chlamtac2018densest}.
More rarely, but intriguingly, there are applications that merge these maximization and minimization views: we are given a set of \emph{desirable}  elements and a separate, disjoint set of \emph{undesirable} elements.
The benefit of $\mathcal{S}' \subseteq \mathcal{S}$ comes from the former; the cost comes from the latter.
(Thus, the groundset contains both ``good'' and ``bad'' elements, and we seek to cover many good and few bad elements.)

Applications abound; regrettably, it is all too common for that which we seek to be bundled with that which we seek to avoid.
Given a social network with a set of ``target'' nodes and a set of ``nontarget'' nodes, choose a set of influencers/{\allowbreak}seed nodes to reach many targets and few nontargets~\cite{8621973}.
For example, a company may advertise a discount code through various influencers, seeking to inform many new potential customers (targets), while not self-undercutting its pricing for existing customers (nontargets).
Given a set of paths in a graph, choose a subset covering many distinct edges but few distinct nodes (relevant to network reliability~\cite{1498549}).
Given a collection of computer science papers, choose some authors who cover a large number of universities/institutions but a small number of individual conferences (relevant to group fairness~\cite{10.1145/3488560.3498525}).
Given a set of participants who might be hired to collect data by following any of a given set of routes through a collection of points of interest (POIs), choose some routes that cover many POIs but a small number of participants~\cite{8456511}.
Further applications have been identified in, e.g., record linkage in data mining~\cite{bilenko2006adaptive}, online review collation~\cite{nguyen2014review}, and motif identification in computational biology~\cite{10.1093/bioinformatics/btz697}.

Although there has been less attention to this combined view in the literature than to, say, the highly studied \problem{Set Cover} problem, this dual minimization/{\allowbreak}maximization perspective has appeared in certain combinatorial formulations.
In \problem{Red-Blue Set Cover}~\cite{carr2000rbsc,peleg2007}, we choose a subcollection of $\mathcal{S}$ that covers all desired (``blue'') elements while minimizing the number of undesired (``red'') elements that are covered.
This problem was later relaxed into \problem{Positive-Negative Partial Set Cover}~\cite{miettinen2008psc}: the hard requirement of covering all blue elements is dropped, and instead the objective function is generalized to minimize the number of errors in either direction --- i.e., the number of false negatives (uncovered blue elements) plus the number of false positives (covered red elements).

\paragraph*{A novel computational formulation: the Target Approximation Problem (TAP).}

In many of the applications in which this formulation is apt, though, the objective function of \problem{Positive-Negative Partial Set Cover} does not seem to capture the intuitive goal.
For example, take the viral marketing scenario cited above~\cite{8621973}: if an influencer-based discount campaign reaches $b$ potential new customers and $r$ current customers, then, up to constant multipliers, the company's profit from the campaign is determined by $b - r$.
That is, the \emph{margin} between the number of true positives and the number of false positives is what characterizes success, rather than the total number of ``errors.''
(Failing to reach a potential new customer may be a missed opportunity, but it's not a loss.)

In this paper, we formulate a new optimization problem, which we call the \emph{Target Approximation Problem (TAP),} which takes this margin-based view. 
Concretely, we are given a \emph{groundset} $U$ and a \emph{target} $B \subseteq U$ of desirable (``blue'') elements. The remaining groundset elements $R = U - B$ are undesirable (``red'').
We are also given a collection of sets $\mathcal{S} = \set{S_1, \ldots, S_m}$, with each $S_j \subseteq U$.
We seek to (approximately) represent $B$ as the union of some of the given subsets in $\mathcal{S}$, where
there is a benefit to every blue element covered and a cost to every red element.
Formally,
we must find a set $\mathcal{S}' \subseteq \mathcal{S}$ maximizing the \emph{margin} of $\mathcal{S}'$:
\begin{align*} 
    (\text{the number of blue elements in $\mathcal{S}'$}) -
    (\text{the number of red elements in $\mathcal{S}'$}).
\end{align*}
Note that we seek to maximize the difference between the number of true positives (blue elements covered by $\mathcal{S}'$) and the number of false positives (red elements covered by $\mathcal{S}'$);
again, 
there is no benefit to true negatives (uncovered red elements), nor any cost for false negatives (uncovered blue elements).

\paragraph*{The present work.}
We view the formulation of the novel Target Approximation Problem as perhaps our most important contribution.
This paper first formally introduces TAP, and then seeks to address its tractability.
Our first results are unsurprisingly negative: TAP is NP-hard in general (an immediate consequence of the hardness of special cases), and hard to $\Theta(1)$-approximate.
As a result, we focus in this paper on two natural special cases of TAP that may be tractable:
(i) \emph{restricted occurrence,} in which each groundset element occurs in $k$ or fewer of the given subsets; and
(ii) \emph{restricted weight,} in which each subset contains $w$ or fewer elements.
(Intuitively and, often, technically, these restrictions correspond to well-studied special cases of \problem{CNF-SAT} in which clauses only contain a small number of literals or where variables only occur in a limited numbers of clauses.)

\begin{description}
\item[\rm \em Results when either weight $\mathop{\le 2}$ or occurrence $\mathop{\le 2}$.]
  We establish an intriguing interplay between the occurrence and weight constraints: TAP is hard when \emph{either} quantity is nontrivial 
  --- but, perhaps surprisingly, TAP can be solved efficiently when \emph{both} are tightly constrained. 
  Specifically, there is an efficient, exact algorithm when $k = 1$ or $w = 1$,
  but TAP is NP-hard even with a $2$-occurrence or $2$-weight constraint if the other quantity is unrestricted.
  And yet when \emph{both} quantities are constrained, the problem becomes tractable again: an efficient algorithm for $2$-occurrence, $2$-weight TAP emerges from an efficient solution to \problem{Independent Set} on a collection of cycles and paths~\cite{arnborg89:bounded-treewidth-IS}.
\end{description}

\begin{figure}[t]
\newcommand{\drawregion}[4][gray]{ 
  \begin{pgfonlayer}{background}
    \draw[fill=#1, draw=gray!50!black, very thick] (#2,#3) rectangle (\xmax+1,\ymax+1);
  \end{pgfonlayer}
  \node[anchor=north west, below right=0.125pc of {(#2,#3)}, #1, text=black] {#4};
}
\def\xmax{12}
\def\ymax{4}

\centering\begin{tikzpicture}[yscale=-1,x=1.5pc,y=1.5pc]
      \draw[gray!50,step=1] (1,1) grid (\xmax+1,\ymax+1);
      \tikzmath{
        \xbutone = \xmax - 1;
        \ybutone = \ymax - 1;
      }
      \foreach \i in {1, ..., \xbutone} {
        \node[font=\small, rotate=90, at={(\i, 1)}, outer sep=0.25pc, anchor=north west] {$w = \i$} ;
      }
      \node[font=\small, rotate=90, at={(\xmax, 1)}, outer sep=0.25pc, anchor=north west] {$w \ge \xmax$} ;
      \foreach \i in {1, ..., \ybutone} {
        \node[font=\small, at={(1,\i+0.5)}, outer sep=0.25pc, anchor=east] {$k = \i$} ;
      }
      \node[font=\small, at={(1,\ymax+0.5)}, outer sep=0.25pc, anchor=east] {$k \ge \ymax$} ;
      \begin{pgfonlayer}{background}
        \draw[fill=green!25, draw=gray!50!black, very thick] (1,1) -- (1,\ymax+1) -- (2,\ymax+1)
                          -- (2,3) -- (3,3) -- (3,2)
                          -- (\xmax+1,2) -- (\xmax+1,1) -- (1,1);
        \node[anchor=north west,below right=0.5pc of {(1,1)}] {
          polynomial-time solvable
        };

        \drawregion[yellow!50]{ 2}{ 3}{}
        \node (yellow region) at (2.5, 4.25) {};
        \node[at={(-1.5,4.15)}, anchor=north east, align=flush left, text width=7pc, align=flush right, inner xsep=0, outer xsep=0.25pc] (hard but approx) {no exact algorithm; polynomial-time 0.5-approximable};
      \begin{pgfonlayer}{foreground}
        \draw[->, thick] (hard but approx.-15) -| (yellow region);
      \end{pgfonlayer}
      

        \drawregion[gray!25] { 4}{ 2}{}
        \drawregion[gray!25] { 5}{ 2}{$\frac{2011}{2012} \approx 0.99953$-inapprox.}
        \drawregion[gray!25] {12}{ 2}{}
        \drawregion[gray!25] { 3}{ 3}{$\frac{94}{95}     \approx 0.98947$-inapproximable}
        \drawregion[gray!25] { 4}{ 4}{$\frac{47}{48}     \approx 0.97917$-inapproximable}

        \node (brown region) at (12.5, 2.5) {};
        \node[below=0.75pc of {(13,5)}, anchor=north east, align=right, inner xsep=0, overlay] (brown inapprox) {$\frac{667}{668} \approx 0.99850$-inapproximable};
      \begin{pgfonlayer}{foreground}
        \draw[->, thick] (brown inapprox.north -| brown region) -- (brown region);
      \end{pgfonlayer}
      
        \coordinate (two four) at (4.5,2.5);

        \node[above left=0.515pc and 3pc of {(0.5,2)}, anchor=north east, align=right] (two four hard) {no exact algorithm};
       \begin{pgfonlayer}{foreground}
         \draw[<-, thick] (two four) -- +(-0.5pc,-0.5pc) -- (two four hard);
       \end{pgfonlayer}
     \node[font=\Large\tt, fill=white, opacity=0.75, inner sep=0] at (3.5,2.5) {?};
        
      \end{pgfonlayer}
  \end{tikzpicture}

  \caption{Summary of results for $k$-occurrence, $w$-weight TAP: each of the given subsets contains $\mathop{\le}w$ elements, and each element appears in $\mathop{\le}k$ of the given subsets.
  \label{fig:inapproximability-results-summary}
  \label{fig:summary-of-results}}

\end{figure}

\noindent%
Based on these results, we explore two further avenues: first, the boundary between tractability and intractability for small $k$ and $w$, and, second, questions of approximability and inapproximability
(see \Cref{fig:summary-of-results} for a summary):
\begin{description}

\item[\rm \em Results for $2$-weight TAP: hard for $k \ge 3$, but efficiently $0.5$-approximable.]
  The hardness result follows from known hardness for \problem{Vertex Cover} even for low-degree graphs~\cite{greenlaw95:cubic,alimonti97:_hardn_approx_probl_cubic_graph,gjs76}.
%
  Although an efficient exact solution for $2$-weight TAP is thus unlikely, we are able to give an efficient $0.5$-approximation algorithm. 
  This is our most technically involved result, based on careful analysis of the greedy algorithm for \problem{Unweighted Set Cover},
  Specifically, we leverage the results of Parekh~\cite{parekh1988setcover} and Slav{\'\i}k~\cite{slavik1996tight} that lower bound the number of elements covered by the $i$th iteration of the greedy algorithm for \problem{Set Cover}. 
\item[\rm \em Results for $2$-occurrence TAP: hard even when $w = 4$.] We adapt the hardness results for \problem{Vertex Cover} for low-degree graphs to show that $2$-occurrence TAP's hardness endures even with a weight-$4$ constraint.
  (Note that $2$-occurrence, $3$-weight TAP's tractability remains open; this is an interesting unresolved question that we leave for future work. See \Cref{sec:open-questions}.)
\item[\rm \em Inapproximability results for $5$-weight TAP and $3$-occurrence TAP.
  ]
  Finally, we derive concrete inapproximability results (ranging from 0.99953 to 0.97917) for the case in which $k \ge 3$ or $w \ge 5$, even with constraints on the other quantity. 
  Our proofs make use of the known hardness of \problem{$k$DM-$k$} ($k$-dimensional matching with at most $k$ occurrences of each element)~\cite{chlebik03:inapprox} and \problem{$a$-OCC-Max-$2$-SAT} (\problem{Max-$2$-SAT} where variables are limited to occurring in $a$ clauses)~\cite{bermankarpinski1999inapprox}.
\end{description}

\paragraph*{Other related work.}
\label{sec:related}%
To the best of our knowledge, the Target Approximation Problem is a novel formulation, but it has some aspects that echo classical combinatorial problems.
We previously described the connection to the (blue) maximization problem and the (red) minimization problem:
\problem{Max-$k$-Cover}, where we must maximize the number of covered blue elements given a ceiling $k$ on the number of sets~(e.g., \cite{36582,garey-and-johnson,10.1145/285055.285059,hochbaum1998analysis,khuller1999budgeted}), and the much-less-well-studied (and seemingly much harder)
\problem{Min-$k$-Union}, where we must minimize the number of covered red elements given a floor $k$ on the number of chosen sets~\cite{chlamtac2018densest,chlamtavc2017minimizing,10.1007/978-981-96-1090-7_3}.

Some variations have relaxed these problems with broadly similar motivations to ours: e.g., in \problem{$k$-Partial Set Cover}, we must cover at least $k$ groundset elements (though not necessarily all of them) with the smallest number of sets~\cite{gandhi2004approximation,BARYEHUDA2001137,SLAVIK1997251}; see also~\cite{DIMANT2025114891}.
As previously detailed, perhaps the closest matches are \problem{Red-Blue Set Cover}~\cite{carr2000rbsc,peleg2007} and \problem{Positive-Negative Partial Set Cover}~\cite{miettinen2008psc}, both of which incorporate the dual desired/{\allowbreak}undesired perspective.
There are known approximability and inapproximability results for both problems.
The difference in objective functions, though, means that these results yield few direct implications for TAP.

Although the problem itself is very different, in some ways TAP is most reminiscent of \problem{Correlation Clustering}~\cite{bansal2002correlation}, in which we are asked to cluster the nodes of a graph whose given $\pm 1$-weighted edges represent the reward/{\allowbreak}penalty for grouping the corresponding endpoints in the same cluster.
The intuitive similarity stems from the fact that \problem{Correlation Clustering} does not specify the number of clusters: an algorithm is free to choose to create a small number of large clusters (earning credit for many included $+1$ edges but paying cost for simultaneously including more $-1$ edges), or a larger number of smaller clusters (earning fewer $+1$s but also incurring fewer $-1$s).
The parallel in TAP is that an algorithm is free to choose how many, and which, blue elements for which to earn credit for covering, while paying the cost for covering whatever red elements are also incidentally covered.
Indeed, a modification of \problem{Correlation Clustering} called \problem{Overlapping Correlation Clustering} --- in which an individual node can be placed in multiple clusters --- is a generalization of \problem{Positive-Negative Partial Set Cover}~\cite{bonchi2013overlapping}.
The objective functions studied in correlation clustering (minimize errors, or maximize the number of correct answers) are analogous to \problem{Positive-Negative Partial Set Cover}, but a version of \problem{Correlation Clustering} where we maximize the total intracluster margin --- i.e., ignoring intercluster edges, regardless of whether they are labeled $+1$ or $-1$ --- is a closer, interesting relative of our problem.



\section{Preliminaries}
\label{sec:problem}

As described in \Cref{sec:intro},
we start with a fixed groundset $U$ of $n$ \emph{elements} 
and a specified \emph{target} $B \subseteq U$.
Elements in $B$ are called \emph{blue}; elements in $R = U - B$ are called \emph{red.}
%
We are also given a collection $\mathcal{S} = \set{S_1, \ldots, S_m}$ of $m$ \emph{subsets} of $U$.
We are asked to select $\mathcal{S}' \subseteq \mathcal{S}$, and the set of elements in $\mathcal{S}'$ --- that is, $\bigcup_{S_i \in \mathcal{S}'} S_i$ --- is then compared to $B$.
A blue element in $\mathcal{S}'$ is \emph{correct;} red elements in $\mathcal{S}'$ are \emph{incorrect.}
(Again, elements not in $\mathcal{S}'$ play no role in the objective; there is no benefit to uncovered red elements nor any cost for uncovered blue elements.)


\begin{definition}[Target Approximation Problem {[TAP]}]~
\begin{description}
\item[Input:] A groundset $U$ of $n$ elements, a target $B \subseteq U$, and a collection $\mathcal{S} = \set{S_1, \ldots, S_m}$ of subsets of $U$. 
\item[Output:] A subset $\mathcal{S}' \subseteq \mathcal{S}$ that maximizes the \emph{margin} of $\mathcal{S}'$:
  \begin{align}
    \label{def:margin}
    \textstyle
    \left|\set{i \in B : i \in \bigcup_{S_i \in \mathcal{S}'} S_i}\right|
    -
    \left|\set{i \in U - B : i \in \bigcup_{S_i \in \mathcal{S}'} S_i}\right|.
  \end{align}
\end{description}
To avoid trivialities, we assume that every element appears in at least one subset in $\mathcal{S}$; that $B$ and each $S_j$ are nonempty; and that the subsets in $\mathcal{S}$ are distinct.
\end{definition}
We refer to the quantity in \eqref{def:margin} as the \emph{margin} of the set $\mathcal{S}'$; thus in TAP we seek the subcollection of $\mathcal{S}$ with the largest margin.


By the \emph{weight} of $S_i \in \mathcal{S}$, we mean $|S_i|$.
A subset $S_i$'s \emph{blue-weight} (respectively, \emph{red-weight)} is the number of blue (respectively, red) elements it contains.
We consider TAP with the \emph{$w$-weight} constraint, in which each subset has weight at most $w$.
(Thus the $n$-weight constraint yields the fully general problem.)

By the number of \emph{occurrences} of a groundset element $i$, we mean the number of subsets in $\mathcal{S}$ in which $i$ appears.
We consider TAP with the \emph{$k$-occurrence} constraint, under which each element appears in at most $k$ different subsets.
(Thus the $m$-occurrence constraint yields the fully general problem.)


\section{TAP with Restricted Weight}
\label{sec:solving-w-weight}

We start with special cases of TAP in which each subset contains at most $w$ elements, establishing a sharp transition in hardness based on $w$.
We first show that TAP is efficiently solvable when each subset's weight is $1$ (or, slightly more generally, when there is no subset containing both blue and red elements).
We then establish TAP's hardness with the mildest relaxation of this constraint, even if subsets have a $2$-weight restriction (i.e., with just one red and one blue element). We begin with a useful fact:

\begin{remark}
  \label{claim:delete-pure-red+add-pure-blue}
    Let $\mathcal{S}' \subseteq \mathcal{S}$.
  Let $S^+$ be any subset with zero red-weight, and let $S^-$ be any subset with zero blue-weight.
  Then the margin of $\mathcal{S}' \cup \set{S^+}$ is never worse than the margin of $\mathcal{S}'$, and the margin of $\mathcal{S}' \cup \set{S^-}$ is never better. 
\end{remark}
The tractability of $1$-weight TAP follows from \Cref{claim:delete-pure-red+add-pure-blue} 
(see \Cref{sec:proofs:restricted-weight} for the proof; we simply take all and only those subsets with zero red-weight):
\begin{restatable}
  {theorem}{OneWeightTractable}
  The case of TAP in which every subset has either zero blue-weight or zero red-weight (which includes $1$-weight TAP) is solvable in polynomial time.%
  \label{lemma:1-weight-is-easy}%
  \label{lem:only-angels}%
\end{restatable}%
%
\noindent%
We also use \Cref{claim:delete-pure-red+add-pure-blue} to preprocess TAP instances to simplify their structure.
Specifically, we henceforth assume that there are no zero-red-weight subsets: we might as well take all of them, so we delete them along with deleting from the groundset the elements that they would cover.
Similarly, we assume that there are no zero-blue-weight subsets: we would never take these subsets, so we simply delete them from $\mathcal{S}$.
Thus, from here on, we assume that every subset contains both blue and red elements.

Next, we show that $1$-weight TAP is the limit of tractability: $2$-weight TAP is hard.
%
We say a TAP instance is \emph{one-red} if each given subset contains exactly one red element.
We start with a useful fact:

\begin{restatable}
  {lemma}{OneRedIsRedBlueSetCover}
  Consider a one-red instance. 
  Let $\mathcal{S}'$ be a collection of subsets that fails to cover at least one blue element.
  Then there exists a subset $S_i \notin \mathcal{S}'$ such that the margin of $\mathcal{S}' \cup \set{S_i}$ is no worse than the margin of $\mathcal{S}'$.

    Thus, for any one-red TAP instance, there is an optimal set of subsets that includes every blue element.  Furthermore, given any optimal solution, we can efficiently compute another optimal solution that includes every blue element.%
\label{cor:one-red-TAP-is-red-blue-set-cover}
\end{restatable}%
\noindent%
(See \Cref{sec:proofs:restricted-weight}: adding any one-red subset that contains an uncovered blue element can only help the margin.)
Note that, after preprocessing, any $2$-weight TAP instance satisfies the one-red constraint; thus the conclusions of \Cref{cor:one-red-TAP-is-red-blue-set-cover} apply to $2$-weight TAP.
%
As a result, efficiently solving $2$-weight TAP implies an efficient solution to the problem of finding the set of subsets covering all blue elements while covering as few red elements as possible: that is, the \problem{Red-Blue Set Cover} problem~\cite{carr2000rbsc}. 
Carr et al.'s reduction from \problem{Set Cover} to \problem{Red-Blue Set Cover} can be applied essentially unchanged to show the hardness of $2$-weight TAP: 

\begin{theorem}
  \label{thm:angel-devil-is-hard}
  It is NP-hard to solve $w$-weight TAP for any $w \ge 2$.
\end{theorem}
(For the argument, see \Cref{note:w=2:k=3}, in \Cref{sect:both-constrained}.
That result establishes hardness of TAP for $w \ge 2$ even with an additional constraint on element occurrence.)


\section{TAP with Restricted Occurrence}
\label{sec:solving-k-occurrence}

We next consider occurrence-constrained TAP: each element appears in at most $k$ subsets.
As with the weight constraint, we show a sharp transition in hardness based on $k$: TAP can be solved efficiently when elements occur only once, but TAP is hard if elements can occur even twice.
We again start with a useful fact, and then establish the tractability of $1$-occurrence TAP:
\begin{remark}
  \label{claim:disjoint-subsets-sum-margins}
  Let $\mathcal{S}_1, \mathcal{S}_2 \subseteq \mathcal{S}$ cover nonoverlapping sets of groundset elements.
  Then the margin of $\mathcal{S}_1 \cup \mathcal{S}_2$ is precisely the margin of $\mathcal{S}_1$ plus the margin of $\mathcal{S}_2$.
\end{remark}

\begin{restatable}{theorem}{OneOccurrenceTractable}
  $1$-occurrence TAP is solvable exactly in polynomial time.
  \label{lemma:1-occurrence-is-easy}
\end{restatable}%
\noindent%
(See \Cref{sec:proofs:restricted-occurrence} for the proof, which uses \Cref{claim:disjoint-subsets-sum-margins}.
The algorithm simply takes all and only those subsets that have higher blue weight than red weight.)

As with weight restrictions, though, the tractability of $1$-occurrence TAP melts away with the mildest relaxation:  

\begin{restatable}{theorem}{TwoOccurrenceHard}
  \label{thm:2-occurrence-is-hard}
  It is NP-hard to solve $k$-occurrence TAP for any $k \ge 2$.
\end{restatable}
\begin{proof}[sketch; see \Cref{sec:proofs:restricted-occurrence}]
  Our proof relies on the hardness of \problem{Max $k$-SAT} (hard even for $k=2$~\cite{gjs76}): given a Boolean proposition $\varphi$ in $k$-CNF, find the maximum number of clauses that can simultaneously be satisfied.
%
  Consider a Boolean formula $\varphi$ in $k$-CNF, with $n$ variables $\set{x_1, \ldots, x_n}$ and $m$ clauses $\set{c_1, \ldots, c_m}$. 
  Construct a TAP instance with two subsets $X_{i,T}$ and $X_{i,F}$ corresponding to each variable~$x_i$. 
  There are three categories of elements:
    \begin{itemize}
    \item \emph{$m$ clause elements.}  A blue element corresponding to clause $c$ in $\varphi$ appears in the $k$ subsets that satisfy that clause ($X_{i,T}$ if $x_i$ is in $c$, and $X_{i,F}$ if $\overline{x_i}$ is).
      
    \item \emph{$2n(m+1)$ penalty elements.} Each subset contains $m+1$ unique red elements (found exclusively in that subset).

    \item \emph{$n(m+2)$ reward elements.}  For each variable $x_i$ in $\varphi$, there are $m+2$ blue elements found in both $X_{i,T}$ and $X_{i,F}$ (and in no other subset).
    \end{itemize}
    Note that the TAP instance obeys the $k$-occurrence constraint.

    Together, the penalty and reward elements are designed to ensure that an optimal TAP solution must correspond to a truth assignment --- including one, but only one, of $X_{i,T}$ and $X_{i,F}$.
    Furthermore, we argue that, for any set $\mathcal{S'}$ of subsets that corresponds to a truth assignment $\rho$, the margin of $\mathcal{S'}$ is exactly
    \begin{equation*}
      \label{eq:margin-in-2-occurrence-is-hard}
      (\text{the number of clauses in $\varphi$ satisfied by $\rho$}) + n.
    \end{equation*}
      Thus the optimal set of subsets corresponds to the truth assignment maximizing the number of clauses satisfied in $\varphi$, and furthermore the truth assignment is efficiently computable from the optimal set of subsets.
    \qed
\end{proof}
Later, we will tighten the construction in the proof of \Cref{thm:2-occurrence-is-hard} to yield stronger hardness results, including restricted-weight $2$-occurrence cases of TAP,
and to derive inapproximability results for TAP.


\section{Restricted Weight and Restricted Occurrence}
\label{sect:both-constrained}

We have now shown the hardness of 2-weight TAP and 2-occurrence TAP.  But TAP under the combination of these constraints \emph{can} be solved efficiently:

\begin{theorem}
  \label{thm:angel-devil-2-occurrence-is-easy}
  2-occurrence, 2-weight TAP is solvable exactly in polynomial time.%
\end{theorem}
\begin{proof}
  Consider a $2$-occurrence, $2$-weight TAP instance. 
  By \Cref{lem:only-angels} (and the preprocessing described immediately after), we assume that each subset has exactly one blue and one red element.
  We can solve this instance by turning it into a graph in which each node corresponds to a red element and each edge to a blue element.
  Specifically, $\set{r,b} \in \mathcal{S}$ corresponds to a stub (a half-edge) connecting node $r$ with half of the edge $b$.
  Thus an edge and its endpoints correspond to a pair of subsets with a shared blue element between two distinct red elements.
  (We may assume that there is no unmatched half-edge: if some blue element $b$ occurs only once, then, by \Cref{cor:one-red-TAP-is-red-blue-set-cover}, there exists an optimal solution that includes $b$'s sole subset. Thus 
  we could augment our preprocessing to delete such subsets and both their red and blue elements.)
  
  By the 2-occurrence constraint, each node in this graph has degree at most two; as a result, the graph consists only of connected components that are cycles and paths.
  By \Cref{cor:one-red-TAP-is-red-blue-set-cover}, it suffices to choose a collection of subsets that covers all blue elements while minimizing the number of red elements.
  It follows that an optimal set of subsets corresponds directly to a minimum vertex cover of the graph.
  Because our graph consists of a collection of disjoint cycles and paths, finding a minimum vertex cover is easy: take an alternating set of nodes (i.e., skipping every other node), starting from an arbitrary node in a cycle, or from the penultimate node of a path.
  (See, e.g., \cite{arnborg89:bounded-treewidth-IS}.)
  The resulting set of nodes and their incident edges yields an optimal set of subsets.
  \qed
\end{proof}

Combined with our previous results, then, we know $k$-occurrence, 2-weight TAP is hard for general $k$ but efficiently solvable for $k = 2$, and $2$-occurrence, $w$-weight TAP is hard for general $w$ but efficiently solvable for $w = 2$.
Where are the transition points?
(I.e., what is the smallest $k$ and smallest $w$ for which $k$-occurrence, $2$-weight TAP and $2$-occurrence, $w$-weight TAP are hard?)
We resolve the $2$-weight question: even with a $3$-occurrence restriction, $2$-weight TAP is hard. 
We have not been able to fully answer the $2$-occurrence question, but we establish that $2$-occurrence TAP is hard even with a $4$-weight restriction. 

\begin{theorem}
  \label{note:w=2:k=3}
  $3$-occurrence, $2$-weight TAP is NP-hard.
\end{theorem}
\begin{proof}
  Following Carr et al.'s hardness proof for \problem{Red-Blue Set Cover}~\cite{carr2000rbsc}, we reduce from \problem{Set Cover}.
  (The construction will exactly match that of Carr et al.; \Cref{cor:one-red-TAP-is-red-blue-set-cover} will establish its relevance to TAP.)
  Given a \problem{Set Cover} instance with groundset $U$ and a collection $\mathcal{C}$ of subsets of $U$, construct this TAP instance:
  \begin{itemize}
  \item There are $|U|$ blue elements $\set{b_1, \ldots, \smash{b_{|U|}}}$ and $|\mathcal{C}|$ red elements $\set{r_1, \ldots, \smash{r_{|\mathcal{C}|}}}$.
  \item For each set $S_i \in \mathcal{C}$ and for each $j \in S_i$, define a subset $\set{r_i, b_j}$.
    (Thus there are $|S_i|$ subsets corresponding to $S_i$, and $\sum_i |S_i|$ subsets in total.)
  \end{itemize}
  Let $\mathcal{S'}$ be an optimal set of subsets for this TAP instance.
  By \Cref{cor:one-red-TAP-is-red-blue-set-cover}, we can efficiently augment $\mathcal{S'}$ into an optimal TAP solution that contains every blue element; without loss of generality, then, assume $\mathcal{S'}$ contains all blue elements.
Write $R$ to denote the set of red elements contained in $\mathcal{S'}$.
The margin of $\mathcal{S'}$ is precisely $n - |R|$.
But there is a direct correspondence between sets of red elements and subsets of $\mathcal{C}$, so $R$ translates directly into a set cover for $U$, with cost exactly $|R|$.
By the TAP-optimality of $\mathcal{S'}$, then, $\mathcal{S'}$ contains the smallest possible number of red elements (i.e., sets in $\mathcal{C}$) while including all blue elements --- and thus $R$ is an optimal solution to the \problem{Set Cover} instance.

  \problem{Set Cover} remains hard even when each groundset element occurs in exactly two sets (i.e., \problem{Vertex Cover}).
  When the maximum degree of the graph is $d$ (in which case \problem{Vertex Cover} is sometimes called \problem{VC-$d$}), each blue element (= edge) occurs exactly twice (once per endpoint) and each red element (= node) occurs exactly $d$ times (once per incident edge). Thus the resulting TAP instance is $\max(2,d)$-occurrence and $2$-weight (each subset contains one red and one blue element).
  The theorem follows from the fact that
  \problem{VC-3} is NP-hard~\cite{greenlaw95:cubic,alimonti97:_hardn_approx_probl_cubic_graph,gjs76}
  and the fact that our TAP instance can be constructed efficiently.
  \qed
\end{proof}

\begin{restatable}{theorem}{FourWeightTwoOccurrenceHard}
  \label{note:w=4:k=2}
  $2$-occurrence, $4$-weight TAP is NP-hard.
\end{restatable}%
\begin{proof}[sketch; see \Cref{sec:proofs:both-constrained}]
We adapt the construction from \Cref{note:w=2:k=3}: merge into a single set the $\mathit{degree}(u)$ subsets that include the red element corresponding to node $u$.
This ``merging'' converts $\mathit{degree}(u)$ subsets of weight $2$ into one subset of weight $1 + \mathit{degree}(u)$, but does not affect the optimal margin.
\qed
\end{proof}

\section{Greedy Approximation of 2-Weight TAP}
\label{sect:2-weight:approx}

We now turn from exact algorithms to approximations, and give an efficient algorithm for one-red TAP (with no constraint on occurrences), which includes the 2-weight case.
Specifically, we show that the classic greedy \textsf{Set Cover} algorithm yields a 0.5-approximation for one-red TAP instances.

In a one-red instance, we can speak of a red element ``covering'' a blue element.
Under the one-red constraint, the choice to absorb the cost of any particular red element $r$ then allows us to reap the benefit of any blue element that appears in a subset with it.  So a subset like $\set{r, b_1, b_2}$, with the one-red form, means that $r$ covers $b_1$ and $b_2$.  (It might also cover other blue elements because of a different subset that also contains the same element $r$.)  

In the context of one-red TAP, \GREEDY behaves as follows: in each iteration, it picks the red element that covers the largest number of (currently uncovered) blue elements, repeating until all blue elements are covered.

To analyze the performance of \GREEDY, we appeal to existing literature on the performance of this same algorithm on the \textsf{Set Cover} problem.
The key fact is the following result, due independently to Parekh~\cite{parekh1988setcover} and Slav{\'\i}k~\cite{slavik1996tight}:
\begin{restatable}[Parekh~\cite{parekh1988setcover}, Slav{\'\i}k~\cite{slavik1996tight}]{lemma}{ParekhSlavik}%
  \label{parekh-slavik}%
  Let $m_i$ denote the number of elements covered by the $i$th iteration of \GREEDY applied to \problem{(Unweighted) Set Cover}. Then
\begin{equation} \label{parekh-slavik-equation}
  m_i \ge \left\lceil \frac{n - \sum_{j=1}^{i-1} m_j }{\OPT_{\text{SC}}} \right\rceil,
\end{equation}
where the groundset has size $n$ and the size of the optimal set cover is $\OPT_{\text{SC}}$.
\end{restatable}%
\noindent%
Or, to restate this in the TAP context:
let $m_i$ denote the number of blue elements covered in the $i$th iteration of \GREEDY applied to one-red TAP, and let $\OPT_{\text{SC}}$ denote the number of red elements in the optimal TAP solution that covers all $n$ blue elements. Then (\ref{parekh-slavik-equation}) holds.

We leverage \Cref{parekh-slavik} to prove our main result: \GREEDY is a $\frac12$-approximation for one-red TAP (including $2$-weight TAP).
It turns out that showing that \GREEDY covers ``enough'' elements in its first ``few'' iterations will be sufficient to establish \Cref{thm:greedy-is-a-1/2-approx-for-angel-devil}.
Specifically, we first prove the following lemma (where the value of $K$ formalizes the notion of ``few,'' which, it is worth noting, depends on the instance's structure):

\begin{restatable}{lemma}{GreedyApproximationLemmaOne}
  \label{lem:greedy-lemma-1}
  Let $K \ceq \left \lceil \frac12 (n - \OPT_{\text{SC}}) \right \rceil$.
  Then \GREEDY covers at least $2K$ blue elements in its first $K$ moves.
\end{restatable}%
\begin{proof}[sketch; see \Cref{sec:proofs:greedy}]
  We argue by induction on $k$ that, for any $k \le K$, the first $k$ iterations of \GREEDY cover at least $2k$ blue elements.
  The base case is immediate by \Cref{parekh-slavik} if $\OPT_{\text{SC}} < n$ (and if $\OPT_{\text{SC}} = n$ then $K = 0$ and the lemma is vacuous).
  For the inductive case, we consider two possibilities (which are exhaustive by the inductive hypothesis):
  \begin{description}
  \item[\rm \em Case 1:]
    \GREEDY is ``ahead of schedule,'' having covered at least $2k-1$ blue elements in its first $k-1$ iterations.
    Then \GREEDY need only cover one additional blue element in its $k$th iteration, which it must do (unless the algorithm already terminated before the $k$th iteration).
    
  \item[\rm \em Case 2:]
    \GREEDY is ``right on schedule,'' having covered exactly $2k-2$ blue elements in its first $k-1$ iterations.
    In this case, the lemma requires \GREEDY to cover at least two elements on its current iteration.
    Here, we appeal to \Cref{parekh-slavik}'s lower bound on $m_k$ --- noting that by the definition of Case 2 we know that the summation in the numerator in (\ref{parekh-slavik-equation}) is exactly $2k - 2$ --- and algebraic manipulation (and the fact that $k \le K$) to establish the result.
  \qed
  \end{description}
\end{proof}%
\begin{restatable}{theorem}{GreedyOneHalfApproximationOneRed}
  \label{thm:greedy-is-a-1/2-approx-for-angel-devil}
  \GREEDY is a $\frac12$-approximation for the one-red TAP problem
  (and therefore \GREEDY is a $\frac12$-approximation for 2-weight TAP).
\end{restatable}%
\noindent%
\begin{proof}[sketch; see \Cref{sec:proofs:greedy}]
  \Cref{lem:greedy-lemma-1} implies that, after the first $K$ ``good'' moves (covering at least two blue elements per move, on average), only $n - 2K$ additional iterations are required (each covering at least one of the $\mathop{\le} n - 2K$ as-yet-uncovered blue elements).
  The careful choice of $K$ ensures that covering all $n$ blue elements with only $K + (n - 2K) = n - K$ red elements is within a $\tfrac12$ factor of optimal.
  \qed
\end{proof}
%

The bound on the approximation ratio is tight: there are $2$-weight TAP instances in which \GREEDY achieves only $\smash{\frac12} \cdot \OPT$ (see \Cref{ex:bad-example-for-greedy} in \Cref{sec:proofs:greedy}).

\begin{note}
  Following~\Cref{parekh-slavik}, write $m_i$ to denote the number of elements covered by the $i$th iteration of the greedy algorithm for \problem{Set Cover}.
  Write $g$ to denote the number of iterations before $\GREEDY$ terminates.
  Then we have that $m_1 \ge m_2 \ge \cdots \ge m_g$ (and $\sum_{i=1}^g m_i = n$), by the definition of \GREEDY.
  (Parekh~\cite{parekh1988setcover} observes this fact explicitly.)
  Thus $m_1 = \max_i m_i$.

  If $m_1 = 2$, then, the sequence $m_{1, \ldots, g}$ takes the form $\langle 2, \cdots, 2, 1, \cdots 1 \rangle$: that is, \GREEDY makes a sequence of ``$2$-moves'' (each covering two new elements) followed by a sequence of ``$1$-moves'' (each covering one new element), until it terminates. 
  Case 2 of the proof of \Cref{lem:greedy-lemma-1} arises precisely when $m_1 = m_2 = \cdots = m_{k-1} = 2$; our proof argues that $m_k$ must also equal $2$ (under the assumption that $k \le K$).
  (Case 1 of the proof can arise only if $m_1 \ge 3$.
  This case corresponds to an instance in which \GREEDY at some point ``gets ahead'' of the 2-move pace --- i.e., $m_1 + m_2 + \cdots m_i > 2i$ for some $i$ --- and, because $i < K$, we still had to argue that $m_1 + m_2 + \cdots m_{i+1} \ge 2i+2$.
  But because the $m$ values are nonincreasing, \GREEDY at some point gets ahead of the 2-move pace if and only if its \emph{first} move is ahead of the 2-move pace.)
\end{note}



\section{Inapproximability Results}
\label{sect:inapproximability}

We now turn to inapproximability, even under restrictions on occurrence and weight.
Our results follow from, first, a tightening of \Cref{thm:2-occurrence-is-hard}, and, second, another hardness derivation based on the $k$-dimensional matching problem.
(Proofs are deferred to \Cref{sec:proofs:inapprox}.)

First, we recall a special case of \problem{MAX-$k$-SAT} called \problem{$a$-OCC-MAX-$k$-SAT}, which is \problem{MAX-$k$-SAT} with the further restriction that variables are limited to occurring in only $a$ clauses.
(Note: our inapproximability bounds for TAP are derived from known inapproximability results for \problem{$a$-OCC-MAX-$2$-SAT}~\cite{bermankarpinski1999inapprox}.
Nonetheless, we state our results in terms of the general \problem{$a$-OCC-MAX-$k$-SAT}, so that any new developments in the approximability of restricted-occurrence \problem{MAX-$k$-SAT} can be translated into the context of TAP.)


\begin{restatable}{theorem}{inapproxViaAOCCMaxKSat}
  \label{thm:inapprox-by-max-k-SAT}
  If \problem{$a$-OCC-MAX-$k$-SAT} is hard to approximate within some factor $\alpha < 1$,
  then $(a + 2\cdot\lfloor \frac{a}{2} \rfloor)$-weight $k$-occurrence TAP is hard to $\alpha$-approximate.
\end{restatable}%
\noindent%
The proof follows the approach of \Cref{thm:2-occurrence-is-hard}, but with a tightened construction (using just 1 reward element per variable and 1 penalty element per subset) and more careful bookkeeping. The result
  allows inapproximability results for \problem{$a$-OCC-MAX-$k$-SAT} to carry over to TAP.
  (For details, see \Cref{sec:proofs:inapprox}.)
Known inapproximability results due to Berman and Karpinski~\cite{bermankarpinski1999inapprox} for \problem{3-Occ-Max-2-SAT} and \problem{6-Occ-Max-2-SAT} then imply hardness for certain TAP cases:%
\begin{restatable}{corollary}{inapproxViaAOCCMaxKSatinstantiated}
  \label{cor:inapprox-by-max-k-SAT:instantiated}
  $2$-occurrence, $5$-weight TAP is hard to $\frac{2011}{2012}$-approximate, and
  $2$-occurrence, $12$-weight TAP is hard to $\frac{667}{668}$-approximate.
\end{restatable}%
\noindent
(Note $\frac{2011}{2012} \mathop{\approx} 0.99953$ and $\frac{667}{668} \mathop{\approx} 0.99850$.)
Again, details are in \Cref{sec:proofs:inapprox}.

 We derive a second set of inapproximability results using the \emph{$k$-dimensional matching} problem:  we are given sets $S_1, S_2, \ldots, S_k$, and a collection of $k$-tuples $C \subseteq S_1 \times S_2 \times \cdots \times S_k$; we must find a subcollection of $C$ of the largest possible size so that no element of any $S_i$ appears more than once in the subcollection.
 (Indeed, 3-dimensional matching was one of Karp's original NP-hard problems~\cite{karp1972}.)
 Denote by \problem{MAX-$k$DM-$k$} the $k$-dimensional matching problem with the additional restriction that each element occurs in at most $k$ sets.

\begin{restatable}{theorem}{inapproxViakDMk}
\label{thm:inapprox-by-max-k-DM-k}
If \problem{MAX-$k$DM-$k$} is hard to approximate within some factor $\alpha < 1$,
then $k$-occurrence, $k$-weight TAP is hard to $\alpha$-approximate.
\end{restatable}
\begin{proof}[sketch; see \Cref{sec:proofs:inapprox}]
  Given sets $S_1, S_2, \ldots, S_k$, and $k$-tuples $C \subseteq S_1 \times S_2 \times \cdots \times S_k$, construct a TAP instance with one blue element for each element of $\bigcup_i S_i$ and $k-1$ red elements for each $k$-tuple $c \in C$.
  Build $k$ subsets corresponding to $c$, one for each index $i$, each containing these $k-1$ red elements and the blue element corresponding to the element of $S_i$ found in $c$.
  The resulting TAP instance is $k$-occurrence, $k$-weight.
  Define a ``canonical'' solution for this TAP instance as one that (i) contains either \emph{none} of the $k$ subsets corresponding to each $k$-tuple $c$, or \emph{all} of them, and (ii) never contains the same blue element more than once.
  We argue that there is an optimal canonical solution, and that one can be efficiently constructed from a not-necessarily-canonical optimal solution.
  Finally, we show that an $\alpha$-approximate canonical TAP solution corresponds directly to an $\alpha$-approximate \problem{MAX-$k$DM-$k$} solution.
  \qed
\end{proof}

Known hardness results by Hazan, Safra, and Schwartz~\cite{hazan2003:kDM} and Chleb{\'\i}k and Chleb{\'\i}kov{\'a}~\cite{chlebik03:inapprox} for \problem{MAX-$k$DM-$k$} imply hardness for additional special cases of TAP (again, details are in \Cref{sec:proofs:inapprox}):
\begin{restatable}{corollary}{inapproxViaKDMKinstantiated}
    \label{cor:inapprox-by-max-k-DM-k:instantiated}
    \label{cor:inapprox-by-max-k-DM:instantiated}
    Unconstrained TAP is hard to $c$-approximate, for any constant $c > 0$.
    Further, 
    $3$-occurrence, $3$-weight TAP is hard to $\frac{94}{95}$-approximate ($\mathop{\approx} 0.98947$)
    and
    $4$-occurrence, $4$-weight TAP is hard to $\frac{47}{48}$-approximate ($\mathop{\approx} 0.97917$).
  \end{restatable}%
\noindent%
Note: \Cref{cor:inapprox-by-max-k-DM-k:instantiated}'s bounds are stronger than \Cref{cor:inapprox-by-max-k-SAT:instantiated}'s when both apply; it is only for $2$-occurrence TAP that our tightest bounds come from \Cref{cor:inapprox-by-max-k-SAT:instantiated}.


\section{Future Work}
\label{sec:open-questions}

Our results on TAP --- including hardness, efficient exact algorithms, and efficient approximation algorithms --- are summarized in \Cref{fig:inapproximability-results-summary}.
These leave several interesting open problems, which we  briefly outline here.

\paragraph*{$2$-occurrence, $3$-weight TAP.}
Our most natural open question is the existence of an efficient algorithm for $2$-occurrence, $3$-weight TAP.
Some of the essential structural properties that allowed an efficient solution for $2$-occurrence, $2$-weight TAP --- roughly, that here TAP reduces to a \problem{Vertex Cover} instance in a graph whose connected components are all cycles or paths --- seem to have rough analogs in $2$-occurrence, $3$-weight TAP.
A potentially promising angle is the fact that \problem{Vertex Cover} is tractable in more general families of graphs, such as graphs with bounded treewidth~\cite{arnborg89:bounded-treewidth-IS}.
We harbor some hope that an efficient algorithm for $2$-occurrence, $3$-weight TAP might emerge through this approach, though we have not yet been able to realize that hope (or to find an argument for hardness).

\paragraph*{Generalized (especially for $3$-weight) or improved approximations.}
Thus far, we have been unable to successfully generalize our 0.5-approximation for $2$-weight, arbitrary-occurrence TAP to broader settings.
Still, as mentioned above, the 3-weight constraint in particular appears to leave intact some of the helpful structural properties that enabled our approximation algorithm.
Is it possible to efficiently approximate $3$-weight TAP?

For $2$-weight TAP, the analysis of the greedy algorithm is tight, but perhaps this approximation be improved by modifying to \GREEDY, or through a different algorithm entirely.
One possible approach involves semi-definite programming: the classic Goemans--Williamson SDP for \problem{MAX-$2$-SAT}~\cite{goemans1994maxcut} has some structural similarities to an SDP for \problem{MIN-$2$-SAT}~\cite{avidor2002approximating}; might it be possible to combine them (using the former for blue elements and the latter for red elements) in some way? Or might there be an approach based on a linear program for \problem{MIN-$2$-SAT}~\cite{bertsimas1999dependent,kohli1994minimum}?

Relatedly, there are sizable gaps between our algorithmic and inapproximability results; might these gaps be closed? Specifically, we might hope to establish inapproximability results for the $2$-weight case (with no occurrence restriction).

\paragraph*{Efficient exact algorithms for additional cases of TAP.}
We concentrated in this paper on restricted-occurrence and restricted-weight instances of TAP.
But there are other tractable special cases, and it is an interesting direction to explore other kinds of special cases that admit an efficient algorithm.
For example, we can determine whether an arbitrary TAP instance can be solved with margin $\mathop{\ge} |B| - k$ in time exponential in $k$, by enumerating all size-$k$ subsets of the groundset (and, for each selected subset, testing whether it's possible to cover all unselected blue elements while covering only the selected red elements).
As another example: although one-red instances are generally hard, under the one-red constraint it suffices to find the optimal \problem{Red-Blue Set Cover},
and \problem{Red-Blue Set Cover} is efficiently solvable in certain geometric settings~\cite{madireddy2021geometric,10.1016/j.ipl.2024.106485,CHAN2015380} or (treating subsets as columns and elements as rows in a binary matrix) instances with the consecutive ones property~\cite{dom2008red,chang2010improved}.
We might hope to find broader sets of TAP instances that are efficiently solvable.

\bigskip\noindent%
Finally, and most broadly, we see as a primary contribution of this paper the introduction of a (to the best of our knowledge) novel problem with an understudied \emph{style} of objective function for covering-type problems.
First, we both seek to reap reward (for desirable elements) and to avoid cost (for undesirable elements).
Second, we have the freedom to ignore ``regions'' of the input where the costs outweigh the benefits; our algorithms have the choice of how much to try to explain.
We see TAP as an fascinating combinatorial problem with these two properties.
What other applications can we address, either with TAP itself, or with other, similarly motivated combinatorial problems?

\begin{credits}
\subsubsection{\ackname} 
{This work was supported in part by \censor{Carleton College}.
We are grateful to \censor{Yang Tan} for work during earlier stages of this project.}

\subsubsection{\discintname}
The authors have no competing interests to declare that are
relevant to the content of this article.
\end{credits}

%
%
%
\bibliography{terse}

\begin{thebibliography}{10}
\providecommand{\url}[1]{\texttt{#1}}
\providecommand{\urlprefix}{URL }
\providecommand{\doi}[1]{https://doi.org/#1}

\bibitem{10.1016/j.ipl.2024.106485}
Abidha, V., Ashok, P.: Red blue set cover problem on axis-parallel hyperplanes
  and other objects. Info. Proc. Ltrs.  \textbf{186}(106485) (Aug 2024)

\bibitem{alimonti97:_hardn_approx_probl_cubic_graph}
Alimonti, P., Kann, V.: Hardness of approximating problems on cubic graphs. In:
  Proc.~Italian Conf. Algorithms and Complexity (CIAC) (1997)

\bibitem{arnborg89:bounded-treewidth-IS}
Arnborg, S., Proskurowski, A.: Linear time algorithms for {NP-}hard problems
  restricted to partial $k$-trees. Discrete Applied Mathematics
  \textbf{23}(1),  11–24 (Apr 1989)

\bibitem{avidor2002approximating}
Avidor, A., Zwick, U.: Approximating {MIN $k$-SAT}. In: ISAAC (2002)

\bibitem{bansal2002correlation}
Bansal, N., Blum, A., Chawla, S.: Correlation clustering. In: FOCS (2002)

\bibitem{BARYEHUDA2001137}
Bar-Yehuda, R.: Using homogeneous weights for approximating the partial cover
  problem. J. of Algorithms  \textbf{39}(2),  137--144 (2001)

\bibitem{bermankarpinski1999inapprox}
Berman, P., Karpinski, M.: On some tighter inapproximability results. In: ICALP
  (1999)

\bibitem{bertsimas1999dependent}
Bertsimas, D., Teo, C., Vohra, R.: On dependent randomized rounding algorithms.
  Operations Research Letters  \textbf{24}(3),  105--114 (1999)

\bibitem{bilenko2006adaptive}
Bilenko, M., Kamath, B., Mooney, R.J.: Adaptive blocking: Learning to scale up
  record linkage. In: IEEE ICDM (2006)

\bibitem{bonchi2013overlapping}
Bonchi, F., Gionis, A., Ukkonen, A.: Overlapping correlation clustering.
  Knowledge and Information Systems  \textbf{35},  1--32 (2013)

\bibitem{carr2000rbsc}
Carr, R.D., Doddi, S., Konjevod, G., Marathe, M.: On the red-blue set cover
  problem. In: SODA (2000)

\bibitem{CHAN2015380}
Chan, T.M., Hu, N.: Geometric red–blue set cover for unit squares and related
  problems. Computational Geometry  \textbf{48}(5),  380--385 (2015)

\bibitem{chang2010improved}
Chang, M.S., Chung, H.H., Lin, C.C.: An improved algorithm for the red--blue
  hitting set problem with the consecutive ones property. Info. Proc. Ltrs.
  \textbf{110}(20),  845--848 (2010)

\bibitem{10.1007/978-981-96-1090-7_3}
Chen, H., Chen, L., Ye, S., Zhang, G.: On extensions of
  {Min-$k$-Union$^\star$}. In: COCOON (2024)

\bibitem{36582}
Chierichetti, F., Kumar, R., Tomkins, A.: Max-cover in map-reduce. In: WWW
  (2010)

\bibitem{chlamtac2018densest}
Chlamt{\'a}c, E., Dinitz, M., Konrad, C., Kortsarz, G., Rabanca, G.: The
  densest $k$-subhypergraph problem. SIAM J. Discrete Math.  \textbf{32}(2),
  1458--1477 (2018)

\bibitem{chlamtavc2017minimizing}
Chlamt{\'a}{\v{c}}, E., Dinitz, M., Makarychev, Y.: Minimizing the union: Tight
  approximations for small set bipartite vertex expansion. In: SODA (2017)

\bibitem{chlebik03:inapprox}
Chleb{\'\i}k, M., Chleb{\'\i}kov{\'a}, J.: Inapproximability results for
  bounded variants of optimization problems. In: Proc.~Fundamentals Computation
  Theory (FCT) (2003)

\bibitem{DIMANT2025114891}
Dimant, S.M., Krumke, S.O.: On approximating partial scenario set cover.
  Theoretical Computer Science  \textbf{1023},  114891 (2025)

\bibitem{dom2008red}
Dom, M., Guo, J., Niedermeier, R., Wernicke, S.: Red-blue covering problems and
  the consecutive ones property. J. of Discrete Algorithms  \textbf{6}(3),
  393--407 (2008)

\bibitem{10.1145/285055.285059}
Feige, U.: A threshold of $\ln n$ for approximating set cover. J. of the ACM
  \textbf{45}(4),  634–652 (Jul 1998)

\bibitem{gandhi2004approximation}
Gandhi, R., Khuller, S., Srinivasan, A.: Approximation algorithms for partial
  covering problems. J. of Algorithms  \textbf{53}(1),  55--84 (2004)

\bibitem{gjs76}
Garey, M.R., Johnson, D.S., Stockmeyer, L.: Some simplified {NP-}complete graph
  problems. Theoretical Computer Science  \textbf{1}(3),  237--267 (1976)

\bibitem{garey-and-johnson}
Garey, M.R., Johnson, D.S.: Computers and Intractability: A Guide to the Theory
  of NP-Completeness. W. H. Freeman and Company (1979)

\bibitem{goemans1994maxcut}
Goemans, M.X., Williamson, D.P.: .879-approximation algorithms for {MAX CUT}
  and {MAX 2SAT}. In: STOC (1994)

\bibitem{greenlaw95:cubic}
Greenlaw, R., Petreschi, R.: Cubic graphs. ACM Computing Surveys
  \textbf{27}(4),  471–495 (Dec 1995)

\bibitem{hazan2003:kDM}
Hazan, E., Safra, S., Schwartz, O.: On the complexity of approximating
  $k$-dimensional matching. In: RANDOM/APPROX (2003)

\bibitem{hochbaum1998analysis}
Hochbaum, D.S., Pathria, A.: Analysis of the greedy approach in problems of
  maximum $k$-coverage. Naval Research Logistics (NRL)  \textbf{45}(6),
  615--627 (1998)

\bibitem{8456511}
Huang, P., Zhu, W., Liao, K., Sellis, T., Yu, Z., Guo, L.: Efficient algorithms
  for flexible sweep coverage in crowdsensing. IEEE Access  \textbf{6},
  50055--50065 (2018)

\bibitem{karp1972}
Karp, R.: Reducibility among combinatorial problems. In: Complexity of Computer
  Computations, vol.~40, pp. 85--103. Plenum (1972)

\bibitem{khuller1999budgeted}
Khuller, S., Moss, A., Naor, J.S.: The budgeted maximum coverage problem. Info.
  Proc. Ltrs.  \textbf{70}(1),  39--45 (1999)

\bibitem{kohli1994minimum}
Kohli, R., Krishnamurti, R., Mirchandani, P.: The minimum satisfiability
  problem. SIAM J. Discrete Math.  \textbf{7}(2),  275--283 (1994)

\bibitem{10.1093/bioinformatics/btz697}
Li, Y., Liu, Y., Juedes, D., Drews, F., Bunescu, R., Welch, L.: Set cover-based
  methods for motif selection. Bioinformatics  \textbf{36}(4),  1044--1051 (Sep
  2019)

\bibitem{10.1145/3488560.3498525}
Ma, H., Guan, S., Toomey, C., Wu, Y.: Diversified subgraph query generation
  with group fairness. In: WSDM (2022)

\bibitem{madireddy2021geometric}
Madireddy, R.R., Nandy, S.C., Pandit, S.: On the geometric red-blue set cover
  problem. In: WALCOM. pp. 129--141 (2021)

\bibitem{miettinen2008psc}
Miettinen, P.: On the positive-negative partial set cover problem. Info. Proc.
  Ltrs.  \textbf{108}(4),  219--221 (2008)

\bibitem{nguyen2014review}
Nguyen, T.S., Lauw, H.W., Tsaparas, P.: Review selection using micro-reviews.
  IEEE Transactions on Knowledge and Data Engineering  \textbf{27}(4),
  1098--1111 (2014)

\bibitem{8621973}
Padmanabhan, M.R., Somisetty, N., Basu, S., Pavan, A.: Influence maximization
  in social networks with non-target constraints. In: IEEE BigData (2018)

\bibitem{parekh1988setcover}
Parekh, A.K.: A note on the greedy approximation algorithm for the unweighted
  set covering problem. Tech. rep., Laboratory for Information and Decision
  Systems, Massachusetts Institute of Technology (1988)

\bibitem{peleg2007}
Peleg, D.: Approximation algorithms for the {Label-Cover$_{\text{\rm MAX}}$}
  and {Red-Blue Set Cover} problems. J. of Discrete Algorithms  \textbf{5},
  55--64 (2007)

\bibitem{slavik1996tight}
Slav{\'\i}k, P.: A tight analysis of the greedy algorithm for set cover. In:
  STOC (1996)

\bibitem{SLAVIK1997251}
Slav{\'\i}k, P.: Improved performance of the greedy algorithm for partial
  cover. Info. Proc. Ltrs.  \textbf{64}(5),  251--254 (1997)

\bibitem{1498549}
Yuan, S., Varma, S., Jue, J.: Minimum-color path problems for reliability in
  mesh networks. In: INFOCOM (2005)

\end{thebibliography}
\newpage
\appendix
\section{Deferred Proofs: Restricted Weight and Restricted Occurrence}
\label{sec:proofs:restricted-weight}
\label{sec:proofs:restricted-occurrence}
\label{sec:proofs:both-constrained}


\subsection{TAP with Restricted-Weight Subsets}

\OneWeightTractable*
\begin{proof}
  Consider a TAP instance in which every subset has zero blue-weight or zero red-weight.
  By \Cref{claim:delete-pure-red+add-pure-blue}, we can always include any zero-red-weight subset and exclude any zero-blue-weight subset while only improving our margin.
  Then the obvious algorithm --- select all subsets with zero red-weight (and no subsets with nonzero red-weight) --- computes a set $\mathcal{S}'$ of subsets that includes all blue elements and no red elements.  (Recall that we assume that each element appears at least once, so every blue element appears in some subset, which has zero red-weight by assumption.)  Thus it gets all elements correct and none incorrect, and therefore has a margin of $n$, the maximum possible margin.

  In $1$-weight TAP, each subset $S_i$ has only one element, which cannot be both red and blue, so $S_i$ has either zero blue-weight or zero red-weight.
  \qed
\end{proof}

\OneRedIsRedBlueSetCover*
\begin{proof}
  Let $x$ be any blue element not present in $\mathcal{S}'$.
  By our universal assumption that all elements appear in at least one subset, there exists a subset $S_i$ that contains $x$.
  Adding $S_i$ to $\mathcal{S}'$ yields one more blue element and at most one additional red element: respectively, $x$ and the unique red element in $S_i$.
  (The ``at most'' is because the red element in $S_i$ may already appear elsewhere in $\mathcal{S}'$.)
  Thus the margin of $\mathcal{S}' \cup \set{S_i}$ is no worse than the margin of $\mathcal{S}'$.
  \qed
\end{proof}

\subsection{TAP with Restricted-Occurrence Elements}

\OneOccurrenceTractable*
\begin{proof}
  By \Cref{claim:disjoint-subsets-sum-margins}, we can split any TAP instance into a collection of connected subinstances (i.e., the connected components of the cooccurrence graph, in which the nodes are elements and a pair of elements is joined by an edge if they occur together in any subset), and solve each component independently.
  Consider the margin of each individual subset $S_i$:  that is, $S_i$'s blue-weight minus $S_i$'s red-weight.
  The obvious algorithm is to include the set of subsets with positive margin.
  \Cref{claim:disjoint-subsets-sum-margins} implies that these decisions can indeed be made independently, and thus the resulting set of subsets is optimal.

  Specifically, then, $1$-occurrence TAP is easy: each element appears in only one subset, and thus all pairs of subsets are necessarily disjoint.
  \qed
\end{proof}

\TwoOccurrenceHard*
\begin{proof}
  We will reduce \textsf{Max $k$-SAT} to $k$-occurrence TAP.
  %
%
%
  We begin from an arbitrary instance of the \textsf{Max $k$-SAT} problem --- that is, an arbitrary Boolean formula $\varphi$ in $k$-CNF --- and we construct an instance of $k$-occurrence TAP.  

    Let the $n$ variables of $\varphi$ be $x_1, \ldots, x_n$ and let the $m$ clauses of $\varphi$ be denoted by $c_1, \ldots, c_m$. 
    We construct a groundset and a set of $2m$ subsets, as follows.
    We will describe the subsets with a binary matrix whose columns represent subsets and whose rows represent elements.
    (See \Cref{fig:k-occurrence-hardness-from-k-sat}.)

    There are two subsets corresponding to each variable $x_i$ in $\varphi$: one subset $X_{i,T}$ representing $x_i$ being set to \textsf{True} and one subset $X_{i,F}$ representing $x_i$ being set to \textsf{False}.
    There are three categories of elements (i.e., rows in the matrix), with $m + 2n(m+1) + n(m+2)$ elements in total:
    \begin{itemize}
    \item \emph{$m$ clause elements.}
      For clause $c$ in $\varphi$, the element corresponding to clause $c$ appears in the $k$ subsets that satisfy that clause.
      For each positive literal $x_j$ in clause $c$ in $\varphi$, the subset $X_{j,T}$ contains element $c$; for each negative literal $\overline{x_j}$ in $c$, the subset $X_{j,F}$ contains element $c$.
      All clause elements are blue.
      
    \item \emph{$2n(m+1)$ penalty elements.}
      For each subset, there are $m+1$ elements found exclusively in that subset.
      All penalty elements are red.
      In other words, any chosen subset incurs a cost of $m+1$.

    \item \emph{$n(m+2)$ reward elements.}
      For each variable $x_i$ in $\varphi$, there are $m+2$ elements found in both of the two subsets corresponding to $x_i$ (i.e., $X_{i,T}$ and $X_{i,F}$) and nowhere else.
      All reward elements are blue.
      In other words, choosing  $X_{i,T}$ or $X_{i,F}$ or both generates a benefit of $m+2$.
    \end{itemize}
    (Together, the penalty and reward elements are designed to ensure that an optimal solution to this TAP instance will correspond to a truth assignment --- specifically including one, but only one, of $X_{i,T}$ and $X_{i,F}$.)

    It is straightforward to construct the set of elements, the target, and the set of subsets efficiently from $\varphi$.
    What remains to be shown is that optimally solving TAP yields a solution to \textsf{Max $k$-SAT}.
    Suppose that a set $\mathcal{S}^\ast$ of subsets is an optimal solution to the constructed instance of TAP.
    Observe the following facts:

    \begin{enumerate}
    \item \emph{The instance of TAP satisfies the $k$-occurrence constraint.}

      Because $\varphi$ was in $k$-CNF, each clause of $\varphi$ contains at most $k$ (possibly negated) literals, so each clause element is similarly found in at most $k$ subsets.
      Each reward element appears in two subsets, and each penalty element appears in only one.
      (And $k \ge 2$.)

    \item \emph{For any variable $i$, the set $\mathcal{S}^\ast$ contains exactly one of $X_{i,T}$ and $X_{i,F}$.
        (That is, the set $\mathcal{S}^\ast$ must correspond to a truth assignment for the variables in $\varphi$.)}

      If neither $X_{i,T}$ nor $X_{i,F}$ is present in $\mathcal{S}^\ast$, then we claim that adding $X_{i,T}$ to $\mathcal{S}^\ast$ improves the margin, contradicting optimality:  by adding $X_{i,T}$ we (1) gain $m+2$ reward elements, (2) gain 0 or more clause elements, and (3) incur the cost of $m-1$ penalty elements.
      In total, then, this addition results in a net benefit of $1$ or more.

      And if both $X_{i,T}$ and $X_{i,F}$ are present in $\mathcal{S}^\ast$, then removing $X_{i,T}$ from $\mathcal{S}^\ast$ improves the set's margin, again contradicting optimality:  we remove $m+1$ penalty elements and at most $m$ clause elements, and do not affect the reward elements (which remain in $\mathcal{S}^\ast$ because of $X_{i,F}$), at a net benefit of $1$ or more.
      
    \item \emph{For any set $\mathcal{S}$ of subsets that corresponds to a truth assignment $\rho$ for the variables in $\varphi$, the margin of $\mathcal{S}$ is exactly
        \[
          (\text{the number of clauses in $\varphi$ satisfied by $\rho$}) + n.
        \]
        Thus $\mathcal{S}^\ast$ (the optimal set of subsets) corresponds to the truth assignment that maximizes the number of clauses satisfied in $\varphi$.}

      Consider any set $\mathcal{S}$ of subsets that, for any $i$, contains exactly one of $X_{i,T}$ and $X_{i,F}$.  Then the margin of $\mathcal{S}$ is
      \begin{eqnarray*}
        \\
        && \tikz[baseline, node distance=.3cm]{
           \node[text width=3.5cm, align=center, anchor=north](u) {the number of clauses in $\varphi$ satisfied by $\rho$};
           \node[right=of u] (plus) {$\mathbin{+}$};
           \node[right=of plus] (v) {$(m+2)\cdot n$};
           \node[right=of v] (minus) {$\mathbin{-}$};
           \node[right=of minus] (w) {$(m+1)\cdot n$};
           \node[right=of v] (minus) {$\mathbin{-}$};
           \node[overlay, above=0pc of u, font=\scriptsize, align=center,blue!50!black] {correct clause elements};
           \node[overlay, above=0pc of v, font=\scriptsize, align=center, blue!50!black,text width=1.5cm] {correct reward elements};
           \node[overlay, above=0pc of w, font=\scriptsize, align=center, blue!50!black,text width=1.5cm] {incorrect penalty elements};
           }
        \\[\jot]
        = && (\text{the number of clauses in $\varphi$ satisfied by $\rho$}) + n.
      \end{eqnarray*}
      Thus $\mathcal{S}^\ast$ must be the set of subsets that maximizes this expression --- and, because the only term that varies with $\rho$ is the number of clauses that it satisfies --- $\mathcal{S}^\ast$ must correspond to a truth assignment maximizing the number of satisfied clauses.
    \end{enumerate}
    It follows, then, that the existence of a polynomial-time algorithm that solves $k$-occurrence TAP implies a polynomial-time algorithm for \textsf{Max $k$-SAT}.  Thus $k$-occurrence TAP is NP-hard for any $k \ge 2$.
    \qed
\end{proof}

\begin{figure}[tp]
    \def\mones{\substack{1 \\ 1 \\ \vdots \\ 1}}
    \def\mzeros{\substack{0 \\ 0 \\ \vdots \\ 0}}
    \def\Kones{\substack{1 \\ 1 \\ \vdots \\ 1}}
    \def\Kzeros{\substack{0 \\ 0 \\ \vdots \\ 0}}
    \centering
    \begin{tikzpicture}

     \matrix (C) [
     matrix of math nodes, left delimiter={[},right delimiter={]}, inner xsep=1, inner ysep=5, nodes={minimum width=4pc, minimum height=1.5pc}]
      {
        x_{1,1,t} & x_{1,1,f} & \dots  & x_{1,n,t} & x_{1,n,f} \\
        x_{2,1,t} & x_{2,1,f} & \dots  & x_{2,n,t} & x_{2,n,f} \\
        \vdots    &  \vdots   & \ddots & \vdots    & \vdots \\
        x_{m,1,t} & x_{m,1,f} & \dots  & x_{m,n,t} & x_{m,n,f} \\
        \mones    & \mzeros   & \dots  & \mzeros   & \mzeros \\
        \mzeros   & \mones    & \dots  & \mzeros   & \mzeros \\
        \vdots    & \vdots    & \ddots & \vdots    & \vdots \\
        \mzeros   & \mzeros   & \dots  & \mzeros   & \mones \\
        \Kones    & \Kones    & \dots  & \Kzeros   & \Kzeros \\
        \vdots    & \vdots    & \ddots & \vdots    & \vdots \\
        \Kzeros   & \Kzeros   & \dots  & \Kones    & \Kones \\
      };

      \foreach \start in {5, 6, 8} {
          \draw  ($(+1pc,0) + (C-\start-1.north west)$) edge[|-|, shorten >=0.5pc, shorten <=0.5pc]
          node[midway, left, font=\scriptsize,rotate=90, anchor=south]{$m+1$} 
          ($(+1pc,0) + (C-\start-1.south west)$);
        }
      \foreach \start in {9, 11} {
          \draw  ($(+1pc,0) + (C-\start-1.north west)$) edge[|-|, shorten >=0.5pc, shorten <=0.5pc]
          node[midway, left, font=\scriptsize,rotate=90, anchor=south]{$m+2$}
          ($(+1pc,0) + (C-\start-1.south west)$);
      }

      \begin{pgfonlayer}{background}
        \node[draw,fill=blue!15,inner ysep=-2, fit=(C-1-1) (C-4-5)] (C clause nodes) {};
        \node[draw,fill=red!15,inner ysep=-2, fit=(C-5-1) (C-8-5)] (C bad nodes) {};
        \node[draw,fill=blue!15,inner ysep=-2, fit=(C-9-1) (C-11-5)] (C good nodes) {};
      \end{pgfonlayer}

      \node[above=1pc of C-1-1,rotate=90,anchor=west] {subset $X_{1,T}$};
      \node[above=1pc of C-1-2,rotate=90,anchor=west] {subset $X_{1,F}$};
      \node[above=1pc of C-1-4,rotate=90,anchor=west] {subset $X_{n,T}$};
      \node[above=1pc of C-1-5,rotate=90,anchor=west] {subset $X_{n,F}$};
      
      \node[right=4pc of C clause nodes,rotate=90,anchor=south, text width=6pc] {clause elements\\(all blue)};
      \node[right=4pc of C good nodes,rotate=90,anchor=south, text width=6pc] {reward elements\\(all blue)};
      \node[right=4pc of C bad nodes,rotate=90,anchor=south, text width=6pc] {penalty elements\\(all red)};

    \end{tikzpicture}
    \caption{
      A visualization of the construction in the proof of \Cref{thm:2-occurrence-is-hard}.
      For the Boolean proposition $\varphi$ with variables $x_1, \ldots, x_n$ and clauses $c_1, \ldots, c_m$, write $x_{i,j,T} = 1$ if $x_i$ appears unnegated in $c_j$ (and $x_{i,j,T} = 0$ otherwise) and write $x_{i,j,F} = 1$ if $x_i$ appears negated in $c_j$ (and $x_{i,j,F} = 0$ otherwise).
      Clause elements and reward elements are blue; penalty elements are red.
      \label{fig:k-occurrence-hardness-from-k-sat}}
\end{figure}

\subsection{TAP with Restricted Weight and Restricted Occurrence}

\FourWeightTwoOccurrenceHard*
\begin{proof}
  This result follows directly from \problem{VC-3} construction in the proof of \Cref{note:w=2:k=3} and a general construction that we can apply to any one-red instance.

  Specifically, we can ``collate'' a one-red TAP instance $I$ into an equivalent one-red instance $I'$ in which every red element appears exactly once: for every red element $r$, simply collapse together all subsets that have $r$ as a member.
  $I$ and $I'$ are equivalent in the sense that an optimal (or, indeed, $\alpha$-approximate) solution to one can be efficiently converted into an optimal (or $\alpha$-approximate) solution to the other.
  (The equivalence follows because any set $\mathcal{S}'$ of subsets in $I$ can only be improved by augmenting it to include every unchosen subset that includes a red element that is included in $\mathcal{S}'$.
  Such an expanded set of subsets in $I$ directly corresponds to a set of subsets in $I'$ with precisely the same margin.)

  Applying the ``collation'' operation to the \problem{VC-3} construction in the proof of \Cref{note:w=2:k=3} yields an instance with one subset for each node $u \in V$ in the graph.
  The subset $S_u$ consists of one red element corresponding to $u$ and $\mathit{degree}(u)$ blue elements corresponding to the edges incident to $u$.
  Because $u$ has degree at most $3$, this subset has weight at most $4$.
  Each red (node) element occurs once; each blue (edge) element occurs twice, once per endpoint.
  Thus the collated instance is $2$-occurrence, $4$-weight.
  \qed
\end{proof}
\begin{note}
    A complementary construction ``shatters'' a one-red instance $I$ into an equivalent $2$-weight one-red instance $I''$ whose subsets are all $2$-weight, by splitting each subset $\set{r, b_1, \ldots, b_k}$ into $k$ subsets $\set{r, b_1}, \set{r,b_2}, \ldots, \set{r,b_k}$.
  The three instances $I$ and ``collated $I$'' and ``shattered $I$'' are all equivalent in the sense of \Cref{note:w=4:k=2}.
\end{note}

\section{Deferred Proofs: Greedy Approximation of $2$-Weight TAP}
\label{sec:proofs:greedy}

Consider a one-red instance of TAP with $n$ blue elements. Let $\OPT_{\text{SC}}$ denote the smallest number of red elements in a set of subsets covering all blue elements.
Let $m_i$ denote the number of blue elements covered in the $i$th iteration of \GREEDY.
(Here, \GREEDY means: ``until all blue elements are covered, repeatedly pick the red element that covers [i.e., co-occurs in subsets with] the largest number of uncovered blue elements.'')

This one-red TAP instance is a restatement of an implicit \problem{Set Cover} instance, and thus the following lemma holds (for the values of $n$, $m_i$, and $\OPT_{\text{SC}}$ as listed above):
\ParekhSlavik*
\noindent%
We claimed the following lemma in \Cref{sect:2-weight:approx}; here is the deferred proof.

\GreedyApproximationLemmaOne*
\begin{proof}
  When $\OPT_{\text{SC}} = n$, then $K = 0$ and the claim holds vacuously.
  Otherwise $\OPT_{\text{SC}} < n$. We claim that, for any $k \le K$, the number $\sum_{j=1}^k m_j$ of blue elements covered in the first $k$ iterations is at least $2k$.
  We proceed by induction on $k$.
  For $k = 1$,
  \begin{align*}
    m_1
    &\geq \lceil n / \OPT_{\text{SC}} \rceil
    \tag*{\Cref{parekh-slavik}}\\
    &\geq n / \OPT_{\text{SC}}\\
    &> 1.
      \tag*{$n/\OPT_{\text{SC}} > 1$ because (by assumption) $\OPT_{\text{SC}} < n$}
  \end{align*}
  Because $m_1 > 1$ is an integer, we have $m_1 \geq 2$, and thus the first iteration covers at least $2$ blue elements.
  For $k \ge 2$, either we have already covered all $n$ blue elements or we have not.
  If we have covered all $n$, then we are done because $n \geq 2k$:
  \[ 
    \textstyle
    2k 
    \ \le \ 2 K 
    \ = \ 2 \ceiling{\frac12(n - \OPT_{\text{SC}})}
    \ \le \ 2 \ceiling{\frac12(n - 1)}
    \ \le \ n.
  \]
  If there remain blue elements to cover, we consider the following two cases:
  \begin{description}
  \item[Case 1:] \emph{\GREEDY is ``ahead of schedule'', meaning $\smash{\sum_{j = 1}^{k-1} m_j} \ge 2k-1$.}
        Because there remain blue elements to cover, \GREEDY will cover at least one of them ($m_k \geq 1$). This gives us $\sum_{j = 1}^{k} m_j \ \geq \ (2k - 1) + 1 \ = \ 2k$ as desired.
  
  \item[Case 2:] \emph{\GREEDY is ``right on schedule'', meaning $\sum_{j = 1}^{k-1} m_j = 2k-2$.} In this case, \GREEDY must cover at least two elements on its current iteration.
    First, observe that
  \def\c#1{\mathrlap{#1}\hphantom{n - \textstyle 2 \left[\frac{1}{2}(n - \OPT_{\text{SC}}) + \tfrac12\right]}}
  \begin{align*}
    &n - 2k\\
    &\ge n - 2K \tag*{$k \le K$}\\
    &= n - \textstyle 2 \ceiling{\frac{1}{2}(n - \OPT_{\text{SC}})} \tag*{definition of $K$}\\
    &\geq n - (n - \OPT_{\text{SC}} + 1) \tag*{$2 \ceiling{\frac12 x} \leq x + 1$ for integral $x$} \\
    &\ge \OPT_{\text{SC}} - 1,
  \end{align*}
  so
  \begin{equation}
    \label{n-(2k-2)>=OPT+1}
    \tag{$\ast$}
    {n - (2k - 2)} \ =\  {(n - 2k) + 2} \ \ge\  {(\OPT_{\text{SC}} - 1) + 2} \ = \  {\OPT_{\text{SC}} + 1}.
  \end{equation}
  Therefore,
  \begin{align*}
      \label{eqn:greedy-inductive-proof}
      m_k
      & \ge \left \lceil \frac{n - \sum_{j = 1}^{k-1} m_j}{\OPT_{\text{SC}}} \right \rceil
        \tag*{\Cref{parekh-slavik}}\\
      & = \left \lceil \frac{n - (2k - 2)}{\OPT_{\text{SC}}} \right \rceil
      \tag*{definition of Case 2}
      \\
      &\ge \left \lceil \frac{\OPT_{\text{SC}} + 1}{\OPT_{\text{SC}}} \right \rceil
        \tag*{(\ref{n-(2k-2)>=OPT+1})}        \\
       & = 2.
  \end{align*}
All together, then, we have $\sum_{j = 1}^{k} m_j  \ge (2k-2) + 1 = 2k$, as desired.
\end{description}
The cases are exhaustive by the inductive hypothesis: \GREEDY must have covered at least $2k - 2$ elements in its previous $k-1$ moves.
\qed
\end{proof}

\GreedyOneHalfApproximationOneRed*
\begin{proof}
  Consider any one-red instance of TAP.
  First, recall that our objective function for TAP is the margin of the chosen set of subsets: that is, the number of covered blue elements minus the number of covered red elements.
  Further, recall that, by \Cref{cor:one-red-TAP-is-red-blue-set-cover}, we know that there exists an optimal solution to the TAP instance that covers all blue elements --- and thus covers the smallest number of red elements while doing so.
  In other words, using the implicit \problem{Set Cover} formulation, and writing $n$ as the number of blue elements, we have
  \begin{equation}
    \OPT_{\text{TAP}} = n - \OPT_{\text{SC}}.
    \label{eq:optTAP=n-optSC}
  \end{equation}
  Similarly, because the greedy algorithms for one-red TAP and \problem{Set Cover} make precisely the same choices, we have
  \begin{equation}
    \GREEDY_{\text{TAP}} = n - \GREEDY_{\text{SC}}.
    \label{eq:greedyTAP=n-greedySC}
  \end{equation}
  Now, by \Cref{lem:greedy-lemma-1}, we know that the first $K = \left \lceil \frac12 (n - \OPT_{SC}) \right \rceil$ iterations of \GREEDY cover at least $2K$ blue elements.
  Thus there remain at most $n - 2K$ blue elements to be covered by the remaining moves.
  In the worst case, this takes $n - 2K$ moves, and therefore the number of red elements covered by \GREEDY satisfies
  \begin{equation}
    \GREEDY_{\text{SC}} \le K + (n - 2K) = n - K.
    \label{eq:greedySC<=n-K}
  \end{equation}
  We then have the desired result by algebraic manipulation of the approximation ratio:
  \begin{align*}
      \frac{\GREEDY_{\text{TAP}}}{\OPT_{\text{TAP}}}
      & = \frac{n - \GREEDY_{\text{SC}}}{n - \OPT_{\text{SC}}}
        \tag*{(\ref{eq:optTAP=n-optSC}) and (\ref{eq:greedyTAP=n-greedySC})}\\
      & \ge \frac{n - (n - K)}{n - \OPT_{\text{SC}}}
        \tag*{(\ref{eq:greedySC<=n-K})}\\
      & = \frac{K}{n - \OPT_{\text{SC}}} \\
      & = \frac{\left \lceil \frac12 (n - \OPT_{SC}) \right \rceil}{n - \OPT_{\text{SC}}}
        \tag*{definition of $K$}\\
      & \ge \frac{\frac12(n - \OPT_{\text{SC}})}{n - \OPT_{\text{SC}}}        \\
       & = \tfrac12, 
  \end{align*}
  as desired.
\end{proof}

\begin{example}
  \label{ex:bad-example-for-greedy}
  Consider the following one-red instance of TAP, with six blue elements $\set{b_1, \ldots, b_6}$, six red elements $\set{r_1, \ldots, r_6}$, and the following six subsets:
  \begin{align*}
    A &= \set{b_1, b_5, r_1} \\
    B &= \set{b_2, b_6, r_2} \\
    C &= \set{b_3, b_5, r_3} \\
    D &= \set{b_4, b_6, r_4} \\
    E &= \set{b_5, b_6, r_5} \\
    F &= \set{b_6, r_6}
  \end{align*}
  \GREEDY repeatedly chooses the red element that covers the largest number of (uncovered) blue elements.
  Thus one possible first choice made by \GREEDY is subset $E$: note that each of $\set{A, B, C, D, E}$ covers exactly two uncovered blue elements, one more than $F$, and so \GREEDY could choose $E$ by tie-breaking.
  At this point, we have four uncovered blue elements $\set{b_1, \ldots, b_4}$, and the only way for \GREEDY to cover these four elements is by choosing subsets $\set{A, B, C, D}$ in some order.

  Thus \GREEDY would cover $\set{b_1, b_2, b_3, b_4, b_5, b_6, r_1, r_2, r_3, r_4, r_5}$, for a margin of $1$.
  But choosing subsets $\set{A, B, C, D}$ covers $\set{b_1, b_2, b_3, b_4, b_5, b_6, r_1, r_2, r_3, r_4}$, for a margin of $2$. Thus \GREEDY is in this case a factor of two from optimal.

  While the example above is not $2$-weight, it can be transformed into a $2$-weight instance by reversing the procedure described in \Cref{note:w=4:k=2}.
  This results in the following problem instance, to which the argument above also applies:
  \begin{align*}
    A_1 &= \set{b_1, r_1}  & A_2 &= \set{b_5, r_1}\\
    B_1 &= \set{b_2, r_2}  & B_2 &= \set{b_6, r_2}\\
    C_1 &= \set{b_3, r_3}  & C_2 &= \set{b_5, r_3}\\
    D_1 &= \set{b_4, r_4}  & D_2 &= \set{b_6, r_4}\\
    E_1 &= \set{b_5, r_5}  & E_2 &= \set{b_6, r_5}\\
    F &= \set{b_6, r_6}.
  \end{align*}
\end{example}

\section{Deferred Proofs: Inapproximability}
\label{sec:proofs:inapprox}

Although the theorem that follows is less powerful than \Cref{note:w=4:k=2}, it is a useful warmup for the inapproximability result that follows:
\begin{restatable}{theorem}{TwoOccurrenceFiveWeightHard}
   \label{thm:5weight-2occurrence-is-hard}
   $2$-occurrence, $5$-weight TAP is NP-hard.
 \end{restatable}
\begin{proof}
  Recall that that \problem{3-OCC-Max-2-SAT} is the \problem{Max-2-SAT} problem where variables are further restricted to appear (in positive or negated form) at most three times. This problem is known to be hard~\cite{bermankarpinski1999inapprox}.
  We reduce from \problem{3-OCC-Max-2-SAT} to $2$-occurrence, $5$-weight TAP, using a variation on the construction in \Cref{thm:2-occurrence-is-hard}: namely, we modify the number of penalty and reward elements to have just $1$ reward element per variable and $1$ penalty element per subset.
  The construction is otherwise the same, but the argument is slightly changed: instead of ensuring that the optimal TAP solution \emph{must have} selected exactly one of $\set{X_{i,T}, X_{i,F}}$ for each $i$, we argue that any optimal solution \emph{can be efficiently modified} to create a truth assignment.
  
  To see this claim, take an arbitrary optimal solution $\mathcal{S}^\ast$.
  First, if there is a variable $x_{i}$ such that $X_{i,T} \notin \mathcal{S}^\ast$ and $X_{i,F} \notin \mathcal{S}^\ast$, then adding either of these variables to $\mathcal{S}^\ast$ will not be a loss: the reward and penalty cancel each other out, and any clause elements will only increase the margin.
  Conversely, if there is a variable $x_{i}$ such that both $X_{i,T} \in \mathcal{S}^\ast$ and $X_{i,F} \in \mathcal{S}^\ast$, then removing one will not result in a loss.
  The fact that $x_{i}$ occurs in no more than 3 clauses implies $X_{i,T}$ and $X_{i,F}$ contain no more than 3 clause elements, combined.
  Then, one of $X_{i,T}$ and $X_{i,F}$ contains no more than one clause element, in which case removing it will not decrease the margin (avoiding one penalty element [i.e., a gain of $1$ in the margin], having no effect on the reward elements, and losing at most one clause element).

  The weight of each subset in this TAP instance is at most $5$ (3 clause elements, 1 penalty element, and 1 reward element).  
  Each element appears only once or twice.
  Thus the resulting TAP instance is $5$-weight and $2$-occurrence.
  \qed
\end{proof}

\inapproxViaAOCCMaxKSat*
\begin{proof}
  Consider \problem{$a$-OCC-MAX-$k$-SAT}, a generalization of \problem{$3$-OCC-MAX-$2$-SAT} as deployed in the proof of \Cref{thm:5weight-2occurrence-is-hard}.
  The critical elements of the construction in \Cref{thm:5weight-2occurrence-is-hard} are:
  \begin{itemize}
  \item \emph{For each variable, the number of reward elements equals the number of penalty elements.}
    Doing so ensures that, if we omit both $X_{i,T}$ and $X_{i,F}$ in any TAP solution $\mathcal{S}'$, we can immediately construct another TAP solution $\mathcal{S}''$ that is no worse than $\mathcal{S}'$ and also includes at least one of these two subsets.
    To see this, observe that if both $X_{i,T}$ and $X_{i,F}$ are omitted from $\mathcal{S}'$, then adding, say, $X_{i,T}$ has the following effect on the set's margin: a nonnegative impact on the clause elements ($X_{i,T}$ may include a previously uncovered clause element, or not), a positive impact via reward elements ($\mathop{+}$ the number of reward elements), and a negative impact via penalty elements ($\mathop{-}$ the number of penalty elements).
    If the number of reward elements equals the number of penalty elements, then the later two effects cancel each other out, and $\mathcal{S}' \cup \set{X_{i,T}}$ has no worse a margin than $\mathcal{S}'$.

  \item \emph{The number of penalty elements for a variable $x_i$ is at least $\floor{\frac{a}{2}}$.}
    Having this $\floor{\frac{a}{2}}$ lower bound on the number of penalty elements ensures that we can omit at least one of $X_{i,T}$ and $X_{i,F}$ in an optimal TAP solution, as follows. Because \problem{$a$-OCC-Max-2-SAT} constrains each variable to appear only $a$ times, the number of clause elements for $X_{i,T}$ and $X_{i,F}$, in total, is $a$; thus at least one of $X_{i,T}$ and $X_{i,F}$ has at most $\floor{\frac{a}{2}}$ clause elements. Therefore removing from a set of subsets whichever of $\set{X_{i,T}, X_{i,F}}$ has fewer clause elements (while leaving the other in the set) has the effect of losing at most $\floor{\frac{a}{2}}$ clause elements while also losing at least $\floor{\frac{a}{2}}$ penalty elements (and having no effect on the number of reward elements), for a net change in the margin that is nonnegative.
  \end{itemize}
  Thus we modify the construction in \Cref{thm:5weight-2occurrence-is-hard} to transform a \problem{$a$-OCC-MAX-$k$-SAT} instance into a TAP instance with exactly $\floor{\frac{a}{2}}$ penalty and reward elements per variables. The weight of each subset in the resulting TAP instance is at most $a + 2\cdot\floor{\frac{a}{2}}$, consisting of $a$ clause elements, $\floor{\frac{a}{2}}$ penalty elements, and $\floor{\frac{a}{2}}$ reward elements.  Each clause element appears at most $k$ times; each penalty and reward element appears only once or twice. Thus the resulting instance of TAP obeys the $(a + 2\cdot\lfloor \frac{a}{2} \rfloor)$-weight and $k$-occurrence constraints.

    As before, an $\alpha$-approximation for the constructed TAP instance would imply an $\alpha$-approximation for \problem{$a$-OCC-Max-$k$-SAT}:  the optimal margin in the TAP instance is precisely equal to the number of satisfied clauses in the \problem{$a$-OCC-Max-$k$-SAT} instance, as the reward and penalty values cancel out.
    Therefore, an $\alpha$-approximation for the TAP instance would imply an $\alpha$-approximation for \problem{$a$-OCC-Max-$k$-SAT}.
    \qed
\end{proof}

\inapproxViaAOCCMaxKSatinstantiated*
\begin{proof}
  Known inapproximability results derived by Berman and Karpinski~\cite{bermankarpinski1999inapprox} for \problem{3-Occ-Max-2-SAT} and \problem{6-Occ-Max-2-SAT} imply that TAP is hard to approximate, as claimed:

  \problem{$3$-OCC-MAX-$2$-SAT} is hard to approximate within $\frac{2011}{2012}$~\cite{bermankarpinski1999inapprox}, and $5=(3 +  2\cdot \lfloor \frac{3}{2} \rfloor)$.

  \problem{$6$-OCC-MAX-$2$-SAT} is hard to approximate within $\frac{667}{668}$~\cite{bermankarpinski1999inapprox}, and $12=(6 +  2\cdot \lfloor \frac{6}{2} \rfloor)$.
  \qed
 \end{proof}

\inapproxViakDMk*
\begin{proof}
  From an arbitrary \problem{MAX-$k$DM-$k$} instance, we will construct a corresponding $k$-occurrence, $k$-weight TAP instance.
  We define a canonical type of solution for the resulting TAP instance, and show two facts:
  (i) an arbitrary TAP solution $\mathcal{S}'$ for this instance can be efficiently converted into a canonical TAP solution $\mathcal{S}''$, where $\mathcal{S}''$ has equal or better margin to $\mathcal{S}'$; and
  (ii) there is an efficient mapping between canonical TAP solutions and \problem{MAX-$k$DM-$k$} solutions that preserves the objective function values across the problems.
  Consequently, a TAP solution with margin $\Delta$ can be efficiently converted into an canonical TAP solution with margin at least $\Delta$, which can in turn be efficiently translated into \problem{MAX-$k$DM-$k$} solution containing at least $\Delta$ disjoint $k$-tuples.
  Because (i) and (ii) also imply that $\OPT_{\problem{TAP}} = \OPT_{\problem{DM}}$, if we can efficiently compute a solution to the TAP instance with margin $\Delta \ge \alpha \cdot \OPT_{\problem{TAP}}$ (i.e., an $\alpha$-approximation for $k$-occurrence, $k$-weight TAP), then we can efficiently construct a \problem{MAX-$k$DM-$k$} solution containing $\alpha \cdot \OPT_{\problem{DM}}$ sets (i.e., an $\alpha$-approximation for \problem{MAX-$k$DM-$k$}).
  
  First, we construct a $k$-occurrence, $k$-weight TAP instance from an arbitrary \problem{MAX-$k$DM-$k$} instance.
  Consider an instance of \problem{MAX-$k$DM-$k$} with given sets $S_1, S_2, \ldots, S_k$, and a collection $\mathcal{C} = \set{\smash{m_1, \ldots, m_{|\mathcal{C}|}}}$ of $k$-tuples, where each $k$-tuple $m_j$ is an element of $S_1 \times S_2 \times \cdots \times S_k$.
  From this, we construct a TAP instance consisting of a collection of $k|\mathcal{C}|$ subsets with $\sum_i|S_i| + (k-1)|\mathcal{C}|$ total elements, as follows:
  \begin{itemize}
  \item Define one blue element $b_y$ corresponding to each element $y \in \bigcup_i S_i$.

    Define $k-1$ red elements $r^j_1, \ldots, r^j_{k-1}$ corresponding to each $k$-tuple $m_j \in \mathcal{C}$. 
    
    Thus there are $\sum_i|S_i|$ blue elements and $(k-1)|\mathcal{C}|$ red elements.

  \item 
    For each $k$-tuple $m_j \in \mathcal{C}$, we define $k$ corresponding subsets $$S_j := \set{X_{1,m_j}, \ldots, X_{k,m_j}}.$$ Each subset in $X_{i, m_j} \in S_j$ has precisely one blue element and $k-1$ red elements:
    \begin{itemize}
    \item The subset $X_{i,m_j}$ contains the blue element $b_y$ where $y = (m_j)_i$ --- that is, the blue element corresponding to the $i$th component of $m_j$.
    \item The subset $X_{i,m_j}$ contains the red elements $r^j_1, \ldots, r^j_{k-1}$ corresponding to $m_j$.
    \end{itemize}
    Thus the subsets in $S_j$ all contain the same $k-1$ red elements, but each contains a distinct blue element.
    (See \Cref{example:max-k-DM-k:reduction} for an illustration.)
  \end{itemize}%

%
  It is straightforward to see that the weight of each subset is $k$.
  For occurrence, each red element is found in exactly $k$ subsets, and each blue element can occur in no more than $k$ subsets because of the $k$-occurrence constraint on the given \problem{MAX-$k$DM-$k$} instance.
  Thus the resulting TAP instance is $k$-weight and $k$-occurrence.

  Call \emph{canonical} any solution to this TAP instance that does not contain any two subsets that share the same blue element and further, for every $j$, contains either \emph{all $k$ of the subsets in $S_j$} or it contains \emph{none of the $k$ subsets.} Now we must argue for (i) and (ii).

  For (i), let $\mathcal{S}'$ be any solution to the constructed TAP instance.
  Fix $j$.
  Notice that the set $S_j$ of subsets associated with $m_j$ has $k$ distinct blue elements (a different blue element in each subset) with $k-1$ total red elements (all of which appear in all $k$ subsets).
  Suppose that some but not all of $S_j$ appears in $\mathcal{S}'$ --- i.e., suppose $\mathcal{S}' \cap S_j \neq \emptyset$ but $E \in S_j - \mathcal{S}'$.
  Then $\mathcal{S}'' = \mathcal{S}' \cup \set{E}$ has no worse of a margin than $\mathcal{S}'$: adding $E$ to $\mathcal{S}'$ does not add any red elements (they were already covered by the subsets in $\mathcal{S}' \cap S_j$), and it might add one more blue element to $\mathcal{S}'$ (if it is not already covered by other subsets in $\mathcal{S}'$).

  Applying the above transformation yields a solution $\mathcal{S}'$ to the constructed TAP instance consisting of the union of $S_j$s (i.e., each $S_j \cap \mathcal{S}' = \emptyset$ or $S_j \cap \mathcal{S}' = S_j$).
  Now suppose that $S_j \subseteq \mathcal{S}'$ and $S_{j'} \subseteq \mathcal{S}'$ where there is some index where $(m_j)_i = (m_{j'})_i$ --- that is, where $m_j$ and $m_{j'}$ are not disjoint.
  Then we claim that $\mathcal{S}'' = \mathcal{S}' - S_j$ has no worse of a margin than $\mathcal{S}'$: excising $S_j$ removes $k-1$ red elements from $\mathcal{S}'$ (the $k-1$ red elements shared across the subsets in $S_j$ that appear nowhere else) and removes at most $k-1$ blue elements from $\mathcal{S}'$ (of the $k$ blue elements in $S_j$, at least one remains covered by $S_{j'}$).

  In other words, the above transformation yields a canonical solution $\mathcal{S}''$ to the constructed TAP instance whose margin is no worse than that of $\mathcal{S}'$.
  
  Now, for (ii), observe that any canonical solution to the constructed TAP instance can be written as $\mathcal{S}' = \set{S_j : j \in I}$ for a set of indices $I$ where no single blue element appears in more than one subset in $\mathcal{S}'$.
  The margin of $\mathcal{S}'$ is precisely $|I|$.
  From the perspective of \problem{MAX-$k$DM-$k$}, the set $\set{m_j : j \in I}$ contains $|I|$ element-disjoint $k$-tuples from the given collection $\mathcal{C}$, or, in other words, a \problem{MAX-$k$DM-$k$} solution with objective function value $|I|$. 
  Because the transformations were efficient, the theorem follows.
  \qed
\end{proof}

\inapproxViaKDMKinstantiated*
\begin{proof}
  For the unconstrained version of TAP: a hardness result due to Hazan, Safra, and Schwartz establishes that MAX-$k$DM (with no constraint on the number of occurrences) is hard to $O(k / \ln k)$-approximate~\cite{hazan2003:kDM}.
  Known inapproximability results due to Chleb{\'\i}k and Chleb{\'\i}kov{\'a}~\cite{chlebik03:inapprox} for \problem{MAX-$k$DM-$2$}, and thus for \problem{MAX-$k$DM-$k$}, imply the latter results:
  \problem{MAX-3DM-2} is hard to $\frac{94}{95}$-approximate and
  \problem{MAX-4DM-2} is hard to $\frac{47}{48}$-approximate.
  \qed
\end{proof}

\begin{example}
  \label{example:max-k-DM-k:reduction}
  Consider the 3-dimensional matching instance with the following 3-tuples:
  \begin{align*}
    A &= \tup{1, 5, 9}    & B &= \tup{2, 5, 10} & C &= \tup{2, 7, 11}\\
    D &= \tup{3, 6, 10}   & E &= \tup{3, 8, 12} & F &= \tup{4, 7, 9}\\ 
    && G &= \tup{4, 6, 11}
  \end{align*}
  (Note that the maximum occurrence of any element happens to be $2$.)

  Then, in our construction, we create 12 blue elements $\set{b_1, b_2, \ldots, b_{12}}$, one per element, and we create 14 red elements $\set{r_{A1},r_{A2}, r_{B1},r_{B2}, \ldots, r_{G1},r_{G2}}$, two per $3$-tuple.  We then define 21 subsets, three corresponding to each of $\set{A, B, \ldots, G}$:
  \begin{align*}
    X_{1, A} &= \set{b_1, r_{A1}, r_{A2}} & X_{1, B} &= \set{b_2, r_{B1}, r_{B2}}    & X_{1, C} &= \set{b_2, r_{C1}, r_{C2}}   \\   
    X_{2, A} &= \set{r_{A1}, b_5, r_{A2}} & X_{2, B} &= \set{r_{B1}, b_5, r_{B2}}    & X_{2, C} &= \set{r_{C1}, b_7, r_{C2}}   \\   
    X_{3, A} &= \set{r_{A1}, r_{A2}, b_9} & X_{3, B} &= \set{r_{B1}, r_{B2}, b_{10}} & X_{3, C} &= \set{r_{C1}, r_{C2}, b_{11}} \\
    \\
    X_{1, D} &= \set{b_3, r_{D1}, r_{D2}} & X_{1, E} &= \set{b_3, r_{E1}, r_{E2}}    & X_{1, F} &= \set{b_4, r_{F1}, r_{F2}}   \\   
    X_{2, D} &= \set{r_{D1}, b_6, r_{D2}} & X_{2, E} &= \set{r_{B1}, b_8, r_{B2}}    & X_{2, F} &= \set{r_{F1}, b_7, r_{F2}}   \\   
    X_{3, D} &= \set{r_{D1}, r_{D2}, b_{10}}& X_{3, E} &= \set{r_{E1}, r_{E2}, b_{12}} & X_{3, F} &= \set{r_{F1}, r_{F2}, b_{9}} \\
    \\
    && X_{1, G} &= \set{b_4, r_{G1}, r_{G2}} \\ 
    && X_{2, G} &= \set{r_{G1}, b_6, r_{G2}} \\
    && X_{3, G} &= \set{r_{G1}, r_{G2}, b_{11}}
  \end{align*}
  Each subset has weight three; each element occurs at most three times.
  (Red elements occur exactly three times; each blue element $b_i$ occurs exactly the same number of times that $i$ appears in the 3-dimensional matching instance.)
\end{example}


\end{document}